\documentclass[aps,pra,reprint,floatfix,superscriptaddress, nofootinbib,longbibliography,onecolumn,notitlepage,12pt,tightenlines]{revtex4-1}
\pdfoutput=1
\usepackage[ascii]{inputenc}
\usepackage[bookmarksnumbered,hypertexnames=false,linktocpage]{hyperref}
\usepackage{bbm,braket,microtype,mathrsfs,amsmath,amssymb,xcolor,amsthm,mathtools,graphicx,enumitem,booktabs}
\usepackage[capitalize]{cleveref}
\crefname{section}{Sec.}{Secs.}
\usepackage[T1]{fontenc}
\usepackage[USenglish]{babel}
\usepackage{times}
\usepackage{verbatim}
\usepackage{subfigure}
\newtheorem{thm}{Theorem}\crefname{thm}{Theorem}{Theorems}
\newtheorem{lem}[thm]{Lemma}\crefname{lem}{Lemma}{Lemmas}
\crefname{cor}{Corollary}{Corollaries}
\crefname{rem}{Remark}{Remarks}
\crefname{dfn}{Definition}{Definitions}
\DeclarePairedDelimiter{\abs}{\lvert}{\rvert}
\DeclarePairedDelimiter{\norm}{\lVert}{\rVert}
\DeclareMathOperator{\tr}{Tr}
\DeclareMathOperator{\Tr}{Tr}

\DeclareMathOperator{\TFD}{TFD}
\renewcommand{\vec}{\mathbf}

\newcommand{\vv}{\vec v}
\newcommand{\vw}{\vec w}

\newcommand{\CC}{\mathbb C}
\newcommand{\ZZ}{\mathbb Z}
\newcommand{\RR}{\mathbb R}
\newcommand{\EE}{\mathbb E}
\newcommand{\FF}{\mathbb F}

\newcommand{\ot}{\otimes}

\newcommand{\Hil}{\mathcal H}
\newcommand{\eps}{\varepsilon}
\renewcommand{\epsilon}{\varepsilon}

\newcommand{\hf}{\frac12}
\newcommand{\ssection}[1]{\smallskip\phantomsection\addcontentsline{toc}{section}{#1}\textit{#1.---}}
\newcommand{\ketum}{\ket {\phi^+}}
\newcommand{\braum}{\bra {\phi^+}}
\newcommand{\proj}[1]{\ket{#1}\!\!\bra{#1}}
\newcommand{\ketbra}[2]{\ket{#1}\!\!\bra{#2}}

\newcommand{\eqn}[1]{\begin{equation}\begin{split} #1 \end{split}\end{equation}}
\allowdisplaybreaks[4]

\begin{document}
\title{Quantum Gravity in the Lab: \texorpdfstring{\\}{} Teleportation by Size and Traversable Wormholes}
\author{Adam R. Brown}
\affiliation{Google, Mountain View, CA 94043, USA}
\affiliation{Department of Physics, Stanford University, Stanford, CA 94305, USA}
\author{Hrant Gharibyan}
\affiliation{Department of Physics, Stanford University, Stanford, CA 94305, USA}
\affiliation{Institute for Quantum Information and Matter, Caltech, Pasadena CA 91125, USA}
\author{Stefan Leichenauer}
\affiliation{Google, Mountain View, CA 94043, USA}
\author{Henry W. Lin}
\affiliation{Google, Mountain View, CA 94043, USA}
\affiliation{Physics Department, Princeton University, Princeton, NJ 08540, USA}
\author{Sepehr Nezami}
\affiliation{Institute for Quantum Information and Matter, Caltech, Pasadena, CA 91125, USA}
\affiliation{Google, Mountain View, CA 94043, USA}
\affiliation{Department of Physics, Stanford University, Stanford, CA 94305, USA}
\author{Grant Salton}
\affiliation{Amazon Quantum Solutions Lab, Seattle, WA 98170, USA}
\affiliation{AWS Center for Quantum Computing, Pasadena, CA 91125, USA}
\affiliation{Institute for Quantum Information and Matter, Caltech, Pasadena, CA 91125, USA}
\affiliation{Department of Physics, Stanford University, Stanford, CA 94305, USA}
\author{Leonard Susskind}
\affiliation{Google, Mountain View, CA 94043, USA}
\affiliation{Department of Physics, Stanford University, Stanford, CA 94305, USA}
\author{Brian Swingle}
\affiliation{Condensed Matter Theory Center, Joint Center for Quantum Information and Computer Science, Maryland Center for Fundamental Physics, and Department of Physics, University of Maryland, College Park, MD 20742, USA}
\author{Michael Walter}
\affiliation{Korteweg-de Vries Institute for Mathematics, Institute for Theoretical Physics, Institute for Logic, Language and Computation \& QuSoft, University of Amsterdam, The Netherlands}
\begin{abstract}
With the long-term goal of studying models of quantum gravity in the lab, we propose holographic teleportation protocols that can be readily executed in table-top experiments.
These protocols exhibit similar behavior to that seen in the recent traversable wormhole constructions of~\cite{gao2017traversable,maldacena2017diving}: information that is scrambled into one half of an entangled system will, following a weak coupling between the two halves, unscramble into the other half.
We introduce the concept of \emph{teleportation by size} to capture how the physics of operator-size growth naturally leads to information transmission. The transmission of a signal \emph{through} a semi-classical holographic wormhole corresponds to a rather special property of the operator-size distribution we call \emph{size winding}.
For more general systems (which may not have a clean emergent geometry), we argue that imperfect size winding is a generalization of the traversable wormhole phenomenon. In addition, a form of signaling continues to function at high temperature and at large times for generic chaotic systems, even though it does \emph{not} correspond to a signal going through a geometrical wormhole, but rather to an interference effect involving macroscopically different emergent geometries.
Finally, we outline implementations feasible with current technology in two experimental platforms: Rydberg atom arrays and trapped ions. 
\end{abstract}
\maketitle
\clearpage
\tableofcontents
\clearpage
\nocite{longpaper}

\section{Introduction}
In the quest to understand the quantum nature of spacetime and gravity, a key difficulty is the lack of contact with experiment.
Since gravity is so weak, directly probing quantum gravity means going to experimentally infeasible energy scales. However, a consequence of the holographic principle~\cite{tHooft:1993dmi,Susskind:1994vu} and its concrete realization in the AdS/CFT correspondence~\cite{Maldacena:1997re,Gubser_1998,witten1998anti} (see also~\cite{Banks:1996vh}) is that non-gravitational systems with sufficient entanglement may exhibit phenomena characteristic of quantum gravity. This suggests that we may be able to use table-top physics experiments to probe theories of quantum gravity indirectly. Indeed, the technology for the control of complex quantum many-body systems is advancing rapidly, and we appear to be at the dawn of a new era in physics---the study of quantum gravity in the lab.

One of the goals of this paper is to discuss one way in which quantum gravity can make contact with experiment. We will focus on a surprising communication phenomenon. We will examine a particular entangled state---one that could actually be made in an atomic physics lab---and consider the fate of a message inserted into the system in a certain way. Since the system is chaotic, the message is soon dissolved amongst the constituent parts of the system. The surprise is what happens next. After a period in which the message seems thoroughly scrambled with the rest of the state, the message then abruptly unscrambles, and recoheres at a point far away from where it was originally inserted. The signal has unexpectedly refocused, without it being at all obvious what it was that acted as the lens.

One way to describe this phenomenon is just to brute-force use the Schrodinger equation. But what makes this phenomenon so intriguing is that it has a much simpler explanation, albeit a simple explanation that arises from an unexpected direction~\cite{gao2017traversable}. If we imagine the initial entangled quantum state consists of two entangled black holes, then there is a natural explanation for why the message reappears---it traveled through a wormhole connecting the two black holes! This is a phenomenon that one could prospectively realize in the lab that has as its most compact explanation a story involving emergent spacetime dimensions.

An analogy may be helpful. Consider two people having a conversation, or as a physicist might describe it ``exchanging information using sound waves''. From the point of view of molecular dynamics, it is remarkable that they can communicate at all. The room might contain $10^{27}$ or more molecules with a given molecule experiencing a collision every $10^{-10}$s or so. In such a system, it is effectively impossible to follow the complete dynamics: the butterfly effect implies that a computer would need roughly $10^{37}$ additional bits of precision every time it propagated the full state of the system for one more second. Communication is possible despite the chaos because the system nevertheless possesses emergent collective modes---sound waves---which behave in an orderly fashion.

Our second goal is understanding the emergence of collective gravitational behavior---in a simple scenario---with the language of quantum information science. When quantum effects are important, complex patterns of entanglement can give rise to qualitatively new kinds of emergent collective phenomena. One extreme example of this kind of emergence is precisely the holographic generation of spacetime and gravity from entanglement, complexity, and chaos. In such situations, new physical structures become possible, including wormholes that connect distant regions of spacetime. And like the physics of sound in the chaotic atmosphere of the room, the physics of these wormholes points the way to a general class of quantum communication procedures which would otherwise appear utterly mysterious.

The experimental study of such situations therefore offers a path toward a deeper understanding of quantum gravity. For instance, by probing stringy corrections to the gravitational description, a sophisticated experiment of this type could even provide an alternative handle on the mathematics of string theory.
Another motivation for this work is that many randomized Hamiltonian systems (such as the Sachdev-Ye-Kitaev (SYK) model or certain random matrix models) possess gravitational duals.
Because these model are inherently not fine tuned, their quantum simulations could in principle be easier than many other applications of quantum computers.
Therefore, we believe that quantum experiments simulating such quantum systems have greater potential to be usefully run on near term quantum devices than most other applications that require high accuracy and fine tuning.\\

A companion paper~\cite{longpaper} to this article, by the same authors, provides additional technical details, examples, and further discussion of the physics of holographic teleportation.  

\subsection{The Quantum Circuit}
In this paper, we will consider the quantum circuits shown in \cref{fig:Wormhole_Circuit}. These circuits, which as we will see in \cref{sec:experiment} may be readily created in a laboratory, exhibit the strange recoherence phenomenon we have described.

\begin{figure}
\includegraphics[width=4.5cm]{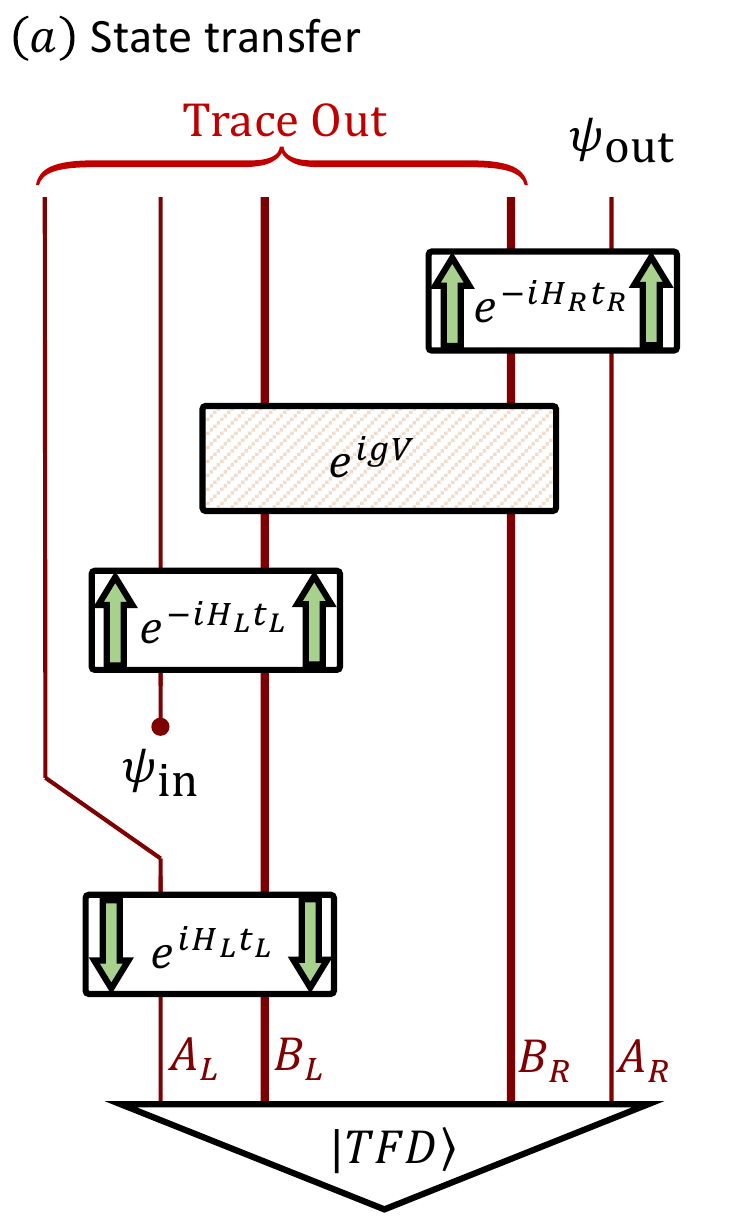}
\qquad\qquad\qquad
\includegraphics[width=4.5cm]{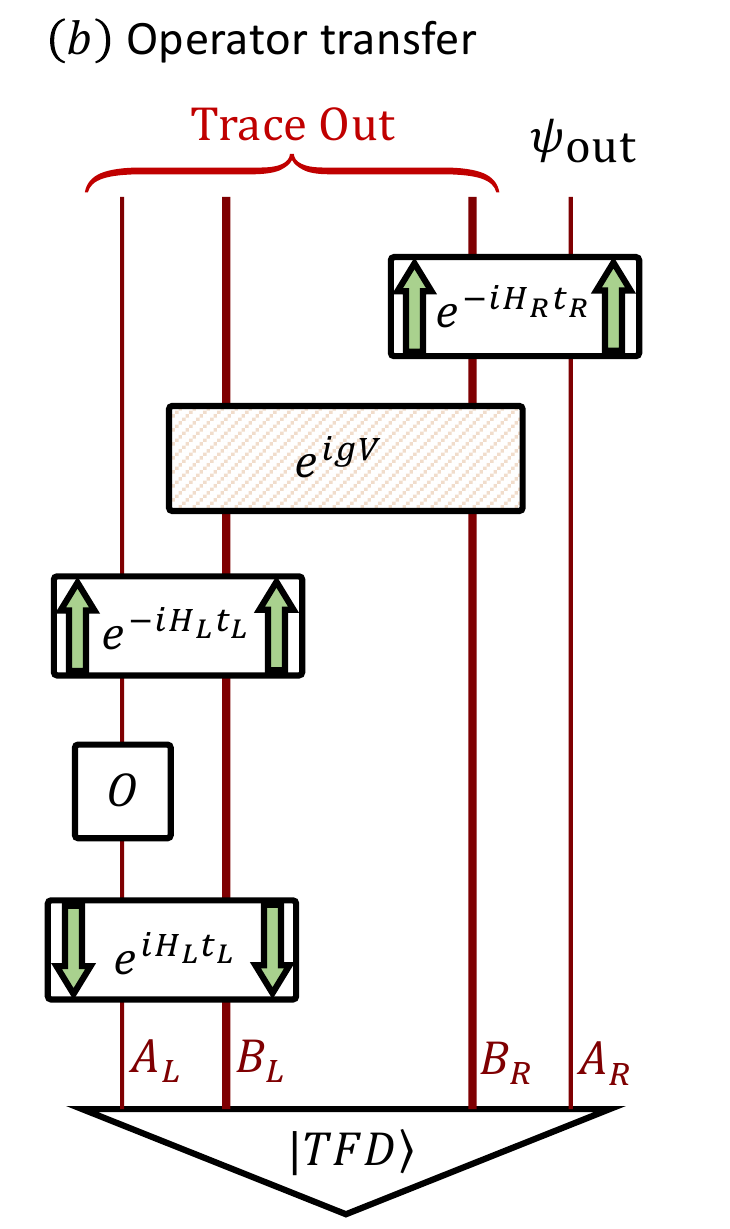}
\caption{The circuits considered in this paper, with $H_L = H_R^T$. Downward arrows indicate acting with the inverse of the time-evolution operator. In both protocols, the goal is to transmit information from the left to the right.
The \textbf{(a) state transfer} protocol calls for us to discard the left message qubits ($A_L$) and replace them with our message $\Psi_{\mathrm{in}}$. The output state on the right then defines a channel applied to the input state.
The \textbf{(b) operator transfer} protocol calls for the operator $O$ to be applied to $A_L$. Based on the choice of operator, the output state on the right is modified, similar to a perturbation-response experiment. \label{fig:Wormhole_Circuit}}
\end{figure}

The circuits act on a $2n$-qubit state. The qubits are divided into $n$ qubits on the left, and $n$ qubits on the right,  subject to Hamiltonians $H$ and $H^T$ respectively, which are assumed to be scrambling~\cite{Sekino_2008, Hayden:2007cs}.
The left and right qubits are initially entangled in the ``thermofield double'' (TFD) state,
\begin{equation}\label{eq:TFD}
    \ket{\TFD}
= \frac{1}{\sqrt{\tr{e^{-\beta H}}}}
    \sum_{j \, \in \,  \textrm{energy levels}}  
    e^{-\beta E_j/2} \, \ket {E_j}_L \ot  \overline{\ket {E_j}}_R,
\end{equation}
where $\beta$ is the inverse temperature, $H \ket {E_j} = E_j \ket {E_j}$, and the bar indicates complex conjugation. We then further partition the systems, labelling $m \ll n$ of the qubits on each side the `message' qubits, and the remaining $n-m$ qubits the `carrier' qubits.

Step one is to bury the message in the left system. First, we evolve all the left qubits `backward in time' by acting with the inverse of the time-evolution operator, $e^{i H t_L}$. Next, we insert the message into the message subsystem of the left qubits. \Cref{fig:Wormhole_Circuit}(a) shows one way to do this---we just throw the existing $m$ qubits away, and replace them with our $m$-qubit message $\Psi_\mathrm{in}$.
\Cref{fig:Wormhole_Circuit}(b) shows another way to do this---keep the $m$ qubits around but act on them with an operator $O$. Next, we evolve the left system `forward in time' using $e^{-i H t_L}$. This forward evolution rapidly scrambles the message amongst the $n$ left qubits.

The next step is to couple the left and right qubits by acting with
\begin{equation*}
\exp{(ig V)},\qquad \text{where }V = \frac{1}{n-m} \hspace{-1mm} \sum_{i \, \in \,  \textrm{carrier qubits} } \hspace{-2mm} Z_i^L Z_i^R \ ,
\end{equation*}
and where $Z_i := (\sigma_z)_i$.
This operation couples each of the left carrier qubits to its mirror image on the right.\footnote{The separation into carrier and message qubits is done solely to clarify our exposition; our results do not change so long as the coupling $V$ acts on sufficiently many qubits (see~\cite{maldacena2017diving}.)}
Finally, we evolve all $n$ of the right qubits `forward in time' using $e^{-i H t_R}$.

It is at this stage that a surprising phenomenon occurs for $t_R \sim t_L$. In the case of state transfer, Figure~\ref{fig:Wormhole_Circuit}(a), the message, so carefully buried on the left, may reappear on the right. In the case of operator transfer, Figure~\ref{fig:Wormhole_Circuit}(b), the action of the operator, so carefully hidden on the left, may become manifest again on the right. The surprise is not that it is information-theoretically possible to recover the message on the right---after all, we coupled the left and right systems with $e^{ig V}$. Instead, the surprise is one of complexity rather than information theory---with the right parameters, we don't need to decode anything, the message just presents itself refocused on the right.
It is not at all obvious how the message made it, and the most surprising fact of all is that the simplest explanation lies in the physics of black holes.

\subsection{Quantum Circuits as Wormholes}

Holography has taught us to think of gravity and spacetime as dual descriptions of non-gravitational quantum systems. In other words, there are chaotic quantum systems, that when looked at differently, can be equivalently described by some type of quantum spacetime with gravity. We tend to call the non-gravitational dual the \emph{boundary}, and the gravitational dual the \emph{bulk}. In a bulk/boundary system, there are phenomena that are more natural in the gravitational bulk side ---those are signatures of the existence of a bulk dual---and phenomena that have natural explanations on the boundary. In this paper and~\cite{longpaper}, our goal is to propose experiments that can be performed on a non-gravitational (boundary) system, which can detect phenomena characteristic of a gravitational dual. Such experiments would enable us to search for signs of the existence of a bulk gravitation dual for non-gravitational quantum systems, and would pave the way to quantum simulations which provide a greater understanding of quantum spacetime.

\begin{figure}
\includegraphics[width=.9\textwidth]{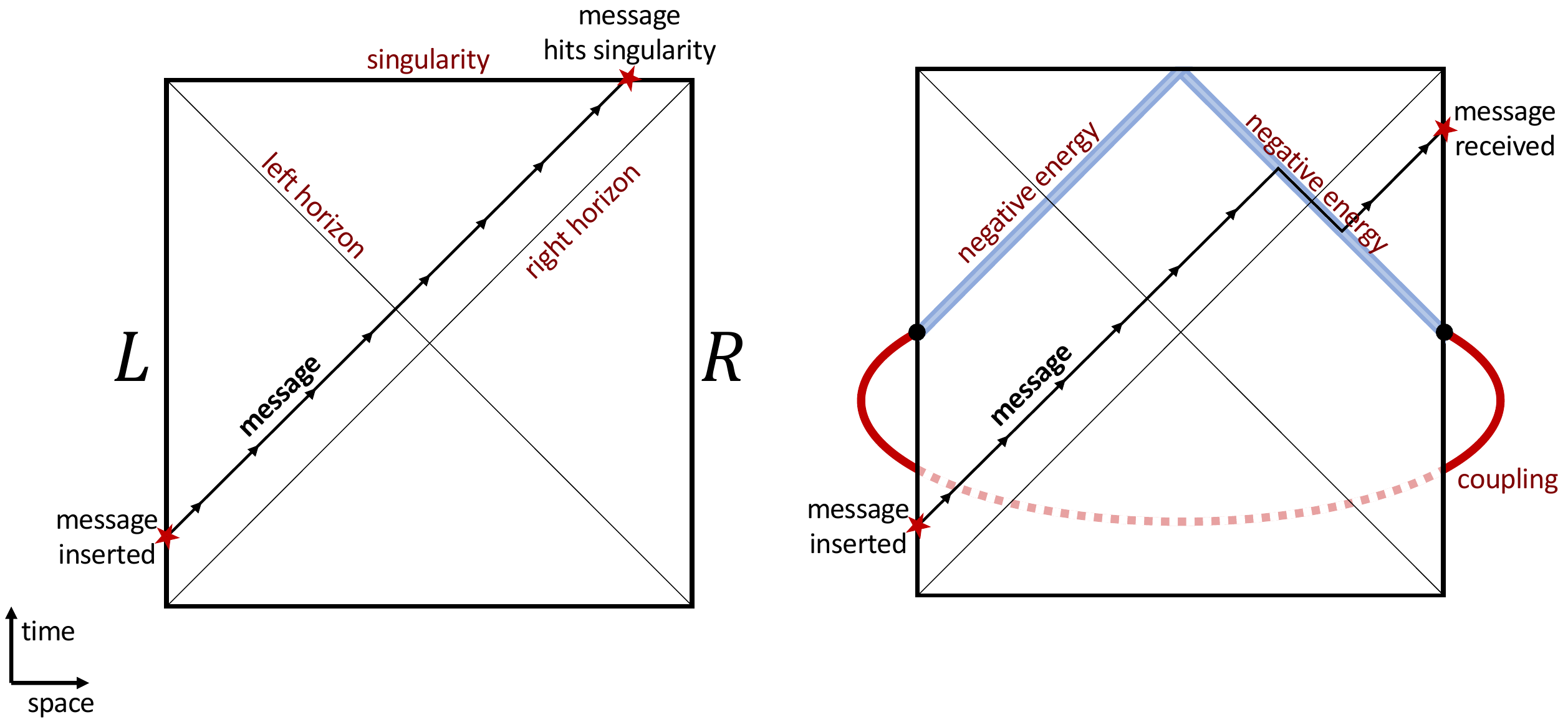}
\caption{Penrose diagram of wormholes.
\textbf{Left}: Without the coupling, a message or particle inserted at early times on the left passes through the left horizon, and hits the singularity (the top line of the diagram).
\textbf{Right}: In the presence of the left-right coupling, the message hits the negative energy shockwave (the thick blue line) created by the coupling. The effect of the collision is to rescue the message from behind the right horizon.}
\label{fig:Penrose_Diagram}
\end{figure}

The AdS/CFT duality is a correspondence between gravitational systems in Anti-de Sitter (AdS) spacetimes and non-gravitational quantum conformal field theories (CFT). A CFT in the thermofield-double state of \cref{eq:TFD} would be dual to the two-sided eternal black hole shown in \cref{fig:Penrose_Diagram}. Such black holes are called `two-sided' because they feature two asymptotic $r=\infty$ regions connected by a wormhole. A pair of observers who jump in from each side may meet before they hit the singularity, but the wormhole is not ``traversable'' since it is not possible to send a signal from the left asymptotic region all the way to the right asymptotic region.

However, in~\cite{gao2017traversable} it was shown how to render such wormholes traversable. A suitably chosen direct coupling between the two sides, which ordinarily do not interact, produces a negative energy shockwave. Negative energy shockwaves impart a time advance to whatever they encounter, and so can rescue a signal that would otherwise be lost to the singularity (\cref{fig:Penrose_Diagram}).

It is this gravitational scattering process that the circuit in Figure~\ref{fig:Wormhole_Circuit} mimics, although the interpretation of the process as traversing a wormhole is \emph{not} valid in general (see Sec.~\ref{sec:holography}). From the gravity perspective, the thermofield double state is used because it exhibits strong left-right correlations (due to the wormhole) that permit negative energy injection. The backward/forward time evolution on the left corresponds to injecting a message in the past on the left. The left-right coupling is the analog of the negative energy shockwave. Finally, the subsequent forward evolution on the right corresponds to allowing the message to travel out to the right boundary where it emerges unscrambled. This process has been called holographic teleportation\footnote{A note on terminology. `Quantum teleportation'~\cite{QTeleportation} refers to using pre-existing entanglement together with classical communication to send a quantum message. If we can do the state-transfer protocol, then we can certainly teleport in the following way~\cite{maldacena2017diving}. Instead of acting with $e^{igV}$, one can simply measure all the left carrier qubits in the $z$-basis, send the \emph{classical} measurement outcomes $z_i \in \{-1,1\}$ over to the right-hand side, and act by $e^{ig\sum_i z_i Z_i^R/(n-m)}$ on the right carrier qubits. This protocol would teleport the message qubits to the right system. It works because the $V$ coupling is classical, i.e., it acts on the left system through a set of commuting operators.} through the wormhole, for which the bulk description is relatively clear. A description of this process in terms of the boundary dynamics has previously been elusive; in this paper, we seek to explain the process from the boundary perspective.

\clearpage\subsection{Summary of Results}
In this paper we identify two distinct mechanisms by which the circuit in \cref{fig:Wormhole_Circuit} can teleport:
\begin{enumerate}
\item High temperature, low capacity teleportation. Holds for times larger than the scrambling time. This mechanism is~\emph{unexpected} from gravity and does not correspond to signals traversing a geometric wormhole. This mechanism only requires that the system dynamics are scrambling, and it is, therefore, applicable to a wide variety of chaotic systems (e.g., random Hamiltonian evolutions, chaotic spin systems, etc.) We also outline experimental proposals for realizing this form of teleportation.
\item Low temperature, high capacity teleportation. Applies near the scrambling time. This regime corresponds to \emph{teleportation through the wormhole}, and it applies to Hamiltonians that have a holographic dual. To understand this mechanism, we introduce the notion of \emph{size winding}, which is an ansatz for the thermal operator near the scrambling time. We explicitly demonstrate size winding in the SYK model, one of the few simple models that are known to have gravitational duals. We thus propose size winding as a general diagnostic of signals traversing a wormhole.
\end{enumerate}

\subsection{Organization of the paper}
In \cref{sec:sizeandmomentum}, we study the circuit of \cref{fig:Wormhole_Circuit} using quantum mechanics, without assuming a holographic dual.
We introduce the notion of \emph{teleportation by size}, we study both mechanisms of teleportation, and we provide general formulas for the fidelity of teleportation.
In \cref{sec:holography} we explain how, in the context of a system with a clean holographic dual, size winding has a direct interpretation in terms of momentum wavefunctions of bulk particles in some appropriate time regime, while in other regimes it need not have a description in terms of particles traversing semiclassical geometrical wormholes.
In \cref{sec:experiment} we discuss concrete experimental realizations of teleportation by size.
The appendices contain proofs of our technical results.

\subsection{Related work}
Other studies of information transfer through traversable wormholes and related notions include~\cite{freivogel2019traversable,bak2019experimental,Bao_2018,bao2019wormhole,gao2019traversable}. In particular, one small-scale experiment with trapped ions has already been carried out~\cite{Landsman_2019} based on~\cite{Yoshida2017,Yoshida2018}. This experiment implemented a probabilistic protocol and a deterministic Grover-like protocol~\cite{Yoshida2017}. In the deterministic case, the circuit in~\cite{Landsman_2019} can be related to our \cref{fig:Wormhole_Circuit}(a) if we specialize to infinite temperature, push the backward time evolution through the thermofield double, and replace~$V$ by a projector onto a Bell pair.

\section{Teleportation by Size}\label{sec:sizeandmomentum}
We base most of our analysis in this paper on size distributions and operator growth, notions heavily studied in connection to holography~\cite{susskind2019complexity,lin2019symmetries,qi2019quantum} and many-body physics~\cite{Nahum_2018,von_Keyserlingk_2018,Xu_2019}---hence the term~\emph{teleportation by size}.
In \cref{subsec:Teleportation_by_Size} we discuss the state transfer protocol.
We will see that state transfer can be done for very generic chaotic quantum systems -- even at infinite temperature. This is the first mechanism of teleportation.
In \cref{subsec:swinding}, we introduce a property of size distributions, called~\emph{size winding}, which we use to explain the second mechanism of teleportation.
Size winding gives a clean mechanism for operator transfer that abstracts the way geometrical wormholes work at the level of the boundary theory (we discuss the latter in \cref{sec:holography} below).
In \cref{subsec:TS_Hamiltonian_General}, we present general bounds on the fidelity of teleportation.

\subsection{Mechanism 1: State Transfer by Size-Dependent Phase}\label{subsec:Teleportation_by_Size}
In this section, we focus on the first mechanism, mentioned above. We study a toy model of state transfer, and we will see that the phenomenon is quite generic. Consider the $2m$-qubit message system $\Hil_{A_L}\ot \Hil_{A_R}$ (see~\cref{fig:Wormhole_Circuit}), and a unitary operator $S = S_{A_LA_R}$ that satisfies
\begin{align}\label{eq:hol_tel_def_cond}
  S\ket{P} =e^{ig' |P|}\ket{P},\quad \text{for all $m$-qubit Paulis }P.
\end{align}
We show in \cref{app:state transfer by size dependent phase} that $S$ maps $\Psi_{\mathrm{in}}\ot\tau$, with $\Psi_{\mathrm{in}}$ an $m$-qubit initial input state and $\tau = I/2^m$ the maximally-mixed state, to
\begin{align}\label{eq:TS_Explicit}
 \Psi_{\mathrm{out}} := \Tr_{A_L}[S(\Psi_{\mathrm{in}}\ot\tau)S^\dagger] = Y^{\ot m} \Delta_\lambda^{\ot m}(\Psi_{\mathrm{in}}) Y^{\ot m},
\end{align}
after tracing out the left subsystem, where $\Delta_\lambda$ is the single qubit depolarizing channel $\Delta_\lambda(\rho) := (1-\lambda)\tau + \lambda\rho$, and $\lambda=(1-\cos(g'))/2$. In pictures,
\vspace{-2mm}
\begin{center}
    \includegraphics[width=8.9 cm]{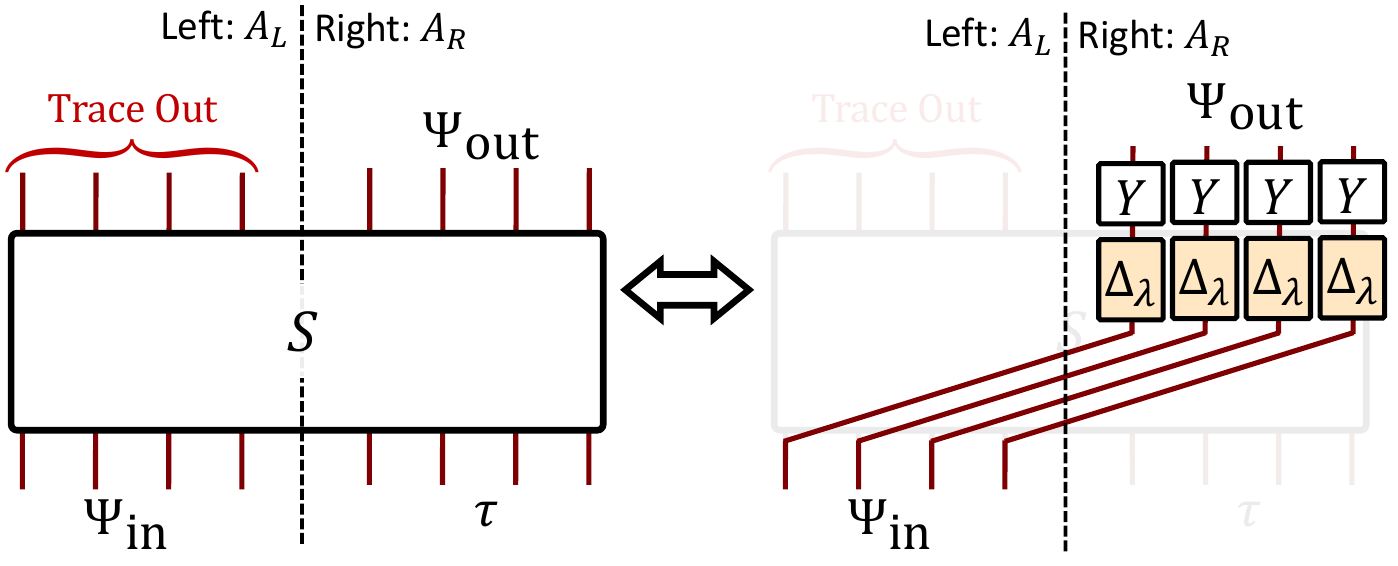}
\end{center}
\vspace{-3mm}
For $g'=\pi$ the state transfer is perfect, whereas for $g'=0$ no signal is sent.

As we show in \cref{fig:Circuit_Equality}, it is natural to look at the $e^{igV}$ coupling between the $L$ and $R$ Hilbert spaces ``sandwiched'' with time evolutions:
\begin{align}\label{eq:sndc}
\left[e^{+iH_Lt} \ot e^{-iH_Rt} \right]e^{igV} \left[e^{-iH_Lt}\ot e^{+iH_Rt}\right].
\end{align}
For many systems of interest, the net effect of the sandwiched coupling on the message subsystems $\Hil_{A_L} \ot \Hil_{A_R}$ is nothing but to approximately implement the unitary $S$ (defined in \cref{eq:hol_tel_def_cond}).
In this way, these systems can achieve high-fidelity state transfer.

\begin{figure}
\includegraphics[width=.65\linewidth]{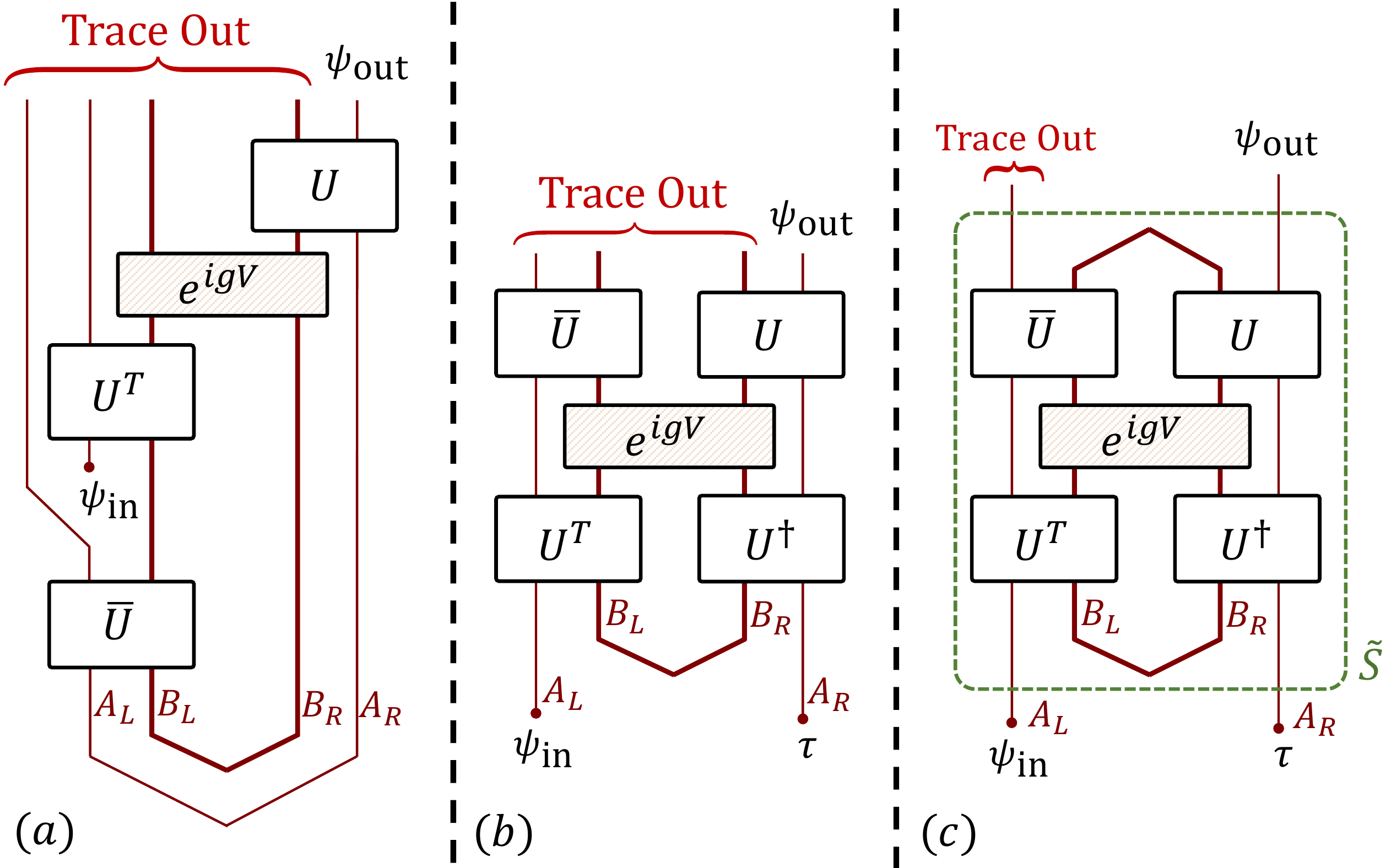}
\caption{{\bf (a):} Infinite-temperature holographic teleportation circuit, with $U=e^{-iHt}$.
{\bf (b):} An equivalent circuit to (a) after circuit manipulations.
{\bf (c):} Result of replacing the trace by projecting the carrier qubits onto $\ketum$.
When the teleportation has high fidelity, this projection has a negligible effect on the final state.
For many systems of interest, the operator~$\tilde S$ enclosed in the dashed rectangle approximately implements the unitary defined in \cref{eq:hol_tel_def_cond} for some appropriate~$g'$.}
\label{fig:Circuit_Equality}
\end{figure}

The simplest case to analyze is when the time evolution~$U = e^{-iH_Rt}$ is described by a Haar random unitary.
In this case, the average of the sandwiched coupling in \cref{eq:sndc} is given by
\begin{align}\label{eq:sandwich haar}
  e^{ig} \phi^+_{LR} + \cos(g/k)^k (I-\phi^+_{LR}) \approx e^{ig} \phi^+_{LR} + (I-\phi^+_{LR}).
\end{align}
up to corrections of order $O(4^{-n})$.
By projecting the carrier qubits onto a maximally entangled state, we thus find that the average of the operator~$\tilde S = \tilde S_{A_LA_R}$ defined in \cref{fig:Circuit_Equality}~(c) over the unitary group is given by $e^{ig} \phi^+_{A_LA_R} + (I-\phi^+_{A_LA_R})$.
When $m=1$, this agrees exactly with \cref{eq:hol_tel_def_cond}, up to a global phase.
In fact, a random instance of $\tilde S$ is close to its average~$S$ with high probability.
For Haar random unitaries, we show this in very strong terms -- these operators are exponentially close (in $n$) to each other in the operator norm, with a probability that is exponentially close (in $n$) to $1$.
See \cref{sec:avg_sec} for a proof.
In contrast, when $m>1$ then the average of~$\tilde S$ no longer coincides with \cref{eq:hol_tel_def_cond}.
Rather, the sandwiched coupling \cref{eq:sandwich haar} acts by applying a constant phase on $\ketum_{A_LA_R}$ and as the identity on all other states in message Hilbert space. This simple operation can be employed to send one qubit, but changing the sign of one state is not a powerful enough operation to send multiple qubits (see \cref{app:Transpose_depo}).\footnote{Note that our results agree with the observations of~\cite{maldacena2017diving} that at late times in the traversable wormhole setup, the commutator of left and right boundary operators acquires an imaginary part, indicating some transmission of information. In~\cite{maldacena2017diving}, the authors observed that at late times the signal should be proportional to $\sin g$, while our calculations indicate a signal proportional to $1-\cos(g)$. This apparent contradiction is resolved in~\cref{app:general_channel}, in which we derive an explicit formula for the output state upon acting by a generic quantum channel on $A_L$. Specifically, in~\cite{maldacena2017diving} it is assumed that one acts on $A_L$ using $e^{i\epsilon O}$, for which we can see from the generic channel in~\cref{app:general_channel} that the response of $\Psi_{\mathrm{out}}$ is proportional to $\sin(g)$, consistent with~\cite{maldacena2017diving}[Eq. (2.20)].}
In fact, no matter what the encoder does at time $-t$ to the system $A_L$ (i.e., acting by a generic channel on $A_L$, which includes the state transfer, operator transfer, and many other protocols), the Holevo information of the full quantum channel from left to right is highly limited. In~\cref{app:general_channel}, we show that it is not possible to send more than $3$ classical bits, and, consequently, $3$ qubits in this way. Moreover, we believe that this bound is a conservative one.

The preceding results hold more generally for 2-designs (which are commonly associated with scrambling and chaos) and can therefore be thought of as modeling the late time behavior of scrambling many-body systems.
In~\cite{longpaper}, we study a variety of other systems in detail, including time evolution with random nonlocal Hamiltonians (GUE or GOE ensembles), $2$-local Brownian circuits, and spin chains.
We show that, at very large times, all models demonstrate the same behavior, but at intermediate times different systems have different physics.

\subsection{Mechanism 2: Size Winding}\label{subsec:swinding}
Consider an observable $O$ and its transpose (in the computational basis) $O^T$ acting at time $-t$ on the left Hilbert space.
Using the definition of the TFD state, this can be expressed as:
\begin{align}\label{eq:sweq1}
\frac{1}{2^{n/2}}O_L^T(-t) \ket {\TFD}_{LR}=(\rho_\beta^{1/2})_RO_R(t)\ketum_{LR},
\end{align}
where $O(t)= e^{iHt} O e^{-iHt}$, $\rho_\beta = e^{-\beta H} ( \text{tr} \, e^{-\beta H})^{-1}$ is the thermal state, and $\ket{\phi^+}$ denotes the maximally entangled state.
The application of $O_L^T(-t)$ should be contrasted with the action of $O_R(t)$ directly on the thermofield double state:
\begin{align}\label{eq:sweq2}
\frac{1}{2^{n/2}}O_R(t)\ket {\TFD}_{LR} = O_R(t)(\rho_\beta^{1/2})_R\ketum_{LR}.
\end{align}
Importantly, the only difference between \cref{eq:sweq1,eq:sweq2} is the order of insertion of~$\rho_\beta^{1/2}$ and~$O(t)$.
Now, expand the operator $\rho_\beta^{1/2} O(t)$ in the Pauli basis as $2^{-n/2}\sum_P c_P P$, where the sum runs over all $n$-qubit Paulis.\footnote{From now on, we suppress the subscripts $L$ and $R$ when there is no confusion.}
Write $\lvert P\rvert$ for the size of an $n$-qubit Pauli operator, i.e., the number of terms not equal to an identity operator.
We define the \emph{winding size distribution}:
\begin{align}\label{eq:size_dist}
  q(l) := \sum_{|P|=l} c_P^2.
\end{align}
The winding size distribution is in contrast to the definition of the conventional size distribution, for which the sum is over the square of the absolute value of $c_P$.
(See~\cite{qi2019quantum} for a proper treatment of fermionic systems.)
The conventional size distribution and the winding size distribution coincide for $\beta = 0$, for which $\rho_\beta^{1/2} O(t)$ is a Hermitian operator and has real expansion coefficients, $c_P \in \RR$.

\emph{Size winding}, in its perfect form, is the following ansatz for the operator wavefunction:
\begin{align*}
    \rho_\beta^{1/2} O(t) =\frac{1}{2^{n/2}} \sum_{P\text{ is an }n\text{-qubit Pauli}} e^{i\alpha |P|/n} r_P P, \quad r_{P} \in \RR.
\end{align*}
The key part of this definition is that the coefficients in the size basis acquire an imaginary phase that is linear in the size of the operators.
If we define $\ket P_{LR}:=P_R\ketum_{LR}$ and assume perfect size winding, then we conclude from the discussion above that
\begin{align}
O_L^T(-t) \ket {\TFD} &= \sum_{P} e^{i\alpha |P|/n} r_P \ket P,\label{eq:erafrleft} \\
O_R(t) \ket {\TFD} &= \sum_{P} e^{-i\alpha |P|/n} r_P  \ket P.\label{eq:erafrright}
\end{align}
Thus, when expressed in the Pauli basis, the difference between the actions of $O^T_L(-t)$ and $O_R(t)$ is given by the ``direction'' of the winding of the phases of the coefficients.

The role of the coupling $e^{igV}$ on a Pauli basis state $\ket P$ is very simple:
it gives a phase of~$-2g/k$ times the number of Pauli $X$ or $Y$ operators acting on the carrier qubits (up to a constant phase).
For typical Pauli operators~$P$, the latter is roughly $2/3$ times the size, hence it follows that
\begin{align*}
  e^{igV} \ket P \approx e^{-i(4/3)g|P|/n}\ket P
\quad \text{up to a constant phase,}
\end{align*}
provided $n\gg m$ (up to a constant phase).
Under the natural hypothesis that the coefficients~$r_P$ only depend on the support of the Pauli operator~$P$, we can similarly show that
\begin{equation}\label{eq:eraf}
  e^{igV} O_L^T(-t) \ket{\TFD}
\approx \sum_{P} e^{i(\alpha - (4/3)g) |P|/n} r_P \ket P
\quad \text{up to a constant phase.}
\end{equation}
\Cref{eq:eraf} illustrates how the weak coupling can transfer a signal from left to right: with a careful choice of~$g$, the action of the coupling unwinds the distribution in \cref{eq:erafrleft} and winds it in the opposite direction to obtain \cref{eq:erafrright}.
This shows that the coupling maps a perturbation of the thermal state of the left system to a perturbation of the right system.
See \cref{app:Twisting} for a precise statement and derivation of this result.

In~\cite{longpaper}, we show that the large-$q$ SYK model exhibits near-perfect size winding and that near-perfect size winding should be present in holographic systems.
Indeed, this is to a large extent nothing but a translation of existing results on two-point functions for traversable wormholes~\cite{maldacena2017diving, qi2019quantum} in the language of size, as we discuss in \cref{sec:holography}.
We will also see that more general size winding, i.e., a size-dependent phase in $q(l)$ that is not necessarily linear in the size, exists in systems without geometric duals.
In fact, we study non-local random Hamiltonian evolution analytically and show that they can weakly transmit a small amount of information in this fashion. See \cref{fig:Summary_of_everything} for a summary of size winding in different scenarios.
\begin{figure}
\includegraphics[width=10 cm]{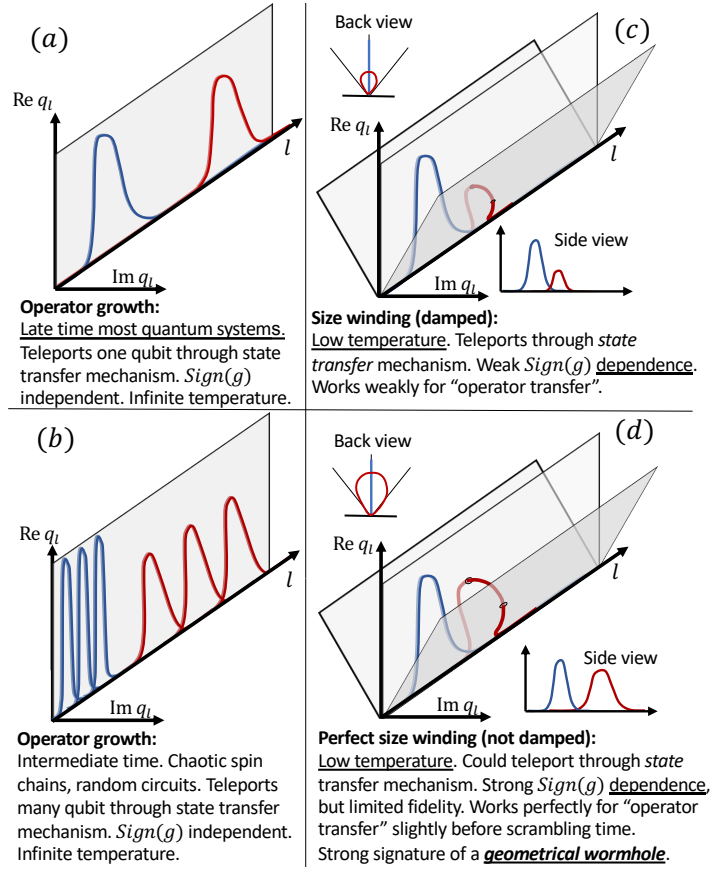}
\caption{A short summary of teleportation by size, discussing different systems, different patterns of operator growth, and consequence of each growth pattern for signal transmission. {\bf Blue:} Initial operator-size distribution. {\bf Red:} Winding size distribution of the time-evolved operator.}
\label{fig:Summary_of_everything}
\end{figure}

\subsection{General Bounds on the Fidelity}\label{subsec:TS_Hamiltonian_General}
In this section, we present general bounds on the entanglement fidelity of the state transfer protocol at arbitrary times and temperatures (\cref{fig:Wormhole_Circuit},~(a)).
For a quantum channel $\mathcal C_{A\rightarrow A}$, the entanglement fidelity~\cite{schumacher1996sending} is given by the overlap between the output and input state when the input is a maximally entangled state between $A$ and an environment $E$ of the same dimension:
$F := \sqrt{\braum_{AE} \mathcal C_{A\rightarrow A}(\phi^{+}_{AE})\ketum_{AE}}$.
Importantly, $F$ lower bounds the average fidelity of the channel over random inputs $\ket \Psi_A$, i.e., $F \leq \EE_{\ket{\Psi}_A}\mathcal F(\Psi_A,\mathcal C(\Psi_A))$.
Motivated by \cref{eq:TS_Explicit}, we take the channel $\mathcal C$ to be the composition of the state transfer protocol with a tensor product of Pauli-$Y$ operators serving as the decoding channel.

Consider a Pauli operator of initial size $l_0$. We assume that Pauli operators with the same initial size $l_0$ have the same generic operator growth, and we denote by~$q_{l_0}$ the corresponding winding size distribution (defined as in \cref{eq:size_dist}).
The central object in our bounds is the Fourier transform of the size distribution, which, for $g^2, gm \ll n$, is equal to the left-right two-point function:
\begin{equation}\label{eq:qtildemaintext}
 \tilde {q}_{l_0}(g) := \sum_{l=0}^n q_{l_0}(l) e^{-i (4g/3)l/n} \approx e^{-ig} \bra{\TFD}  O_R(t) e^{igV} O_L^T(-t) \ket{\TFD}.
\end{equation}
This is proved in \cref{app:two point}.
The fidelity is a difficult quantity to evaluate directly, yet one can still provide strong bounds on the fidelity in terms of the simpler quantity~$F_q$, defined to be
\begin{equation}\label{eq:def F_q main}
    F_q := \abs[\big]{ \sum\limits_{l=0}^{m} \left[N_l/4^m\right](-1)^l\tilde q_l(g) },
    \quad\text{where}\quad N_l = \textstyle\binom{m}{l}3^l.
\end{equation}
We show the following bounds on the entanglement fidelity in terms of $F_q$:
\begin{align}\label{eq:fid_main_text}
 F_q \lessapprox F &\lessapprox F_q + \sum\limits_{l=0}^{m} (N_l/4^m)\sqrt{1-\abs{\tilde q_l(g)}^2}
\end{align}
For local Hamiltonian evolutions, and in a variety of time regimes, $F_q$ can be a good estimate of $F$ as the error term on the right-hand side of \cref{eq:fid_main_text} will be small.\textbf{}
Under the assumption that the thermal state has a narrow size distribution, we can also show that
\begin{align}\label{eq:fid_main_text2}
     F &\lessapprox 1-\frac{1}{4^{m}}(1- F_q).
\end{align}
See \cref{app:Fid_Form} for proofs.
These relations allow us to rigorously bound the entanglement fidelity for various random Hamiltonian and spin chain models in several parameter regimes.

As an example, we can use \cref{eq:fid_main_text} to confirm that random unitary time evolution at infinite temperature should teleport a single qubit as shown in \cref{subsec:Teleportation_by_Size}. At $\beta=0$, $q_{0}$ is peaked at $l=0$ for all times. However, $O(t) = U^\dagger O U$ is a completely random combination of Pauli strings, and thus its size distribution is peaked at $l=(3/4)n$. Hence, $\tilde q_0(g) = 1$, and $\tilde q_1(g) \approx e^{-ig}$. Thus, $F_q = |3e^{-ig}/4 -1/4|=\sqrt{1/4+3/4( 1 - \cos(g))/2 }$. Furthermore, since $|q_l|=1$, we have that $F_q=F$ from \cref{eq:fid_main_text}. Therefore, the channel can teleport with perfect fidelity when $g=\pi$.

\section{The Holographic Interpretation}\label{sec:holography}
The analogy between \cref{fig:Wormhole_Circuit,fig:Penrose_Diagram} is very suggestive, and now we will return to the question of whether the geometric picture is a faithful representation of the physics.
In other words, when can we claim that a message was sent through an emergent geometry?
The teleportation-by-size mechanism we have introduced generalizes the traversable wormhole, and persists even in cases where a fully classical wormhole is not the appropriate description.
In fact, we will see that even in the holographic setting, at very large times the teleportation-by-size paradigm remains valid even when the description in terms of a single semi-classical geometry breaks down.

\subsection{Size and momentum}
The growth of the size of an operator is a basic manifestation of chaos, and is related to a particle falling towards a black hole horizon~\cite{Roberts:2018mnp, Brown:2018kvn, qi2019quantum}. In the context of SYK, or Nearly AdS$_2$ holography, the bulk interpretation of size is particularly sharp~\cite{lin2019symmetries}, which we now review. In the traversable wormhole, the particle crossing the negative-energy shockwave experiences a (null) translation. The shockwave can therefore be interpreted as the generator of this translation, otherwise known as (null) momentum.
The shockwave is a direct consequence of the interaction between the two sides, which in the SYK model is simply the ``size'' operator. Thus, the size operator is simply related to null momentum~\cite{maldacena2017diving,lin2019symmetries,longpaper}.

A more precise argument based on~\cite{lin2019symmetries} can also be given; a detailed version will appear in~\cite{longpaper}.  Readers unfamiliar with Nearly AdS$_2$ may jump to the next section.
The starting point is that for states close to the thermofield double, the operators defined by
\eqn{B = H_R - H_L, \quad E=H_L + H_R + \mu V - E_0 \label{eq:approximate-sym}}
have a simple geometrical action as a Lorentz boost $B$, and as global time translation $E$~\cite{maldacena2018eternal}.
Here $V$ is a sum of operators on both sides $V = \sum_{i=1}^k O_i^L O_i^R$; in the SYK model, the simplest choice would be to take $V \propto i \sum \psi^j_L \psi^j_R$ to be the size operator. The value of $\mu$ and $E_0$ should be tuned so that the TFD is an approximate ground state of $E$, see~\cite{maldacena2018eternal}.
It is then natural to consider the combinations
\def\hf{ \frac{1}{2} }
\eqn{P_\pm = -\hf (E\pm B). }
For our purposes, the important point is that $e^{ia^\pm P_\pm}$ generate a null shift.\footnote{Said more precisely, these generators act as left/right Poincare symmetry generators, which are null shifts at the edge of the Poincare patch.} By choosing the right sign of $a^\pm$, we can shift the particle backwards so that it traverses the wormhole. Now notice that \eqn{-P_+ = H_R +  \mu V/2, \quad -P_- = H_L + \mu V/2.}
The remarkable feature of this formula is that the action of $P_\pm$ is exceedingly simple on the left/right Hilbert space (equivalently, on one-sided operators), since we can ignore $H_L$ or $H_R$. For operators on the left (right) side, the amount of $P_+$ ($P_-$) momentum inserted is just given by the size, up to some normalization.

This in turn implies that the size wavefunction of a one-sided operator $O$ (e.g., the components of $O$ in a basis of operators organized by size) is dual to the momentum-space wavefunction of the particle created by $O$. The Fourier transform of the momentum wavefunction is then related to the ``position'' of the particle in the bulk, where ``position'' here means the AdS$_2$ coordinate conjugate to null momentum. Furthermore, the action of the two-sided coupling $e^{igV}$ in the traversable wormhole protocol simply shifts the position of the particle, allowing the particle to potentially exit the black hole.

The upshot is that in a holographic setting, we can clearly see that the winding of the size distribution is related to the location of the particle, e.g., whether the particle is inside or outside of the black hole horizon.
The case of imperfect winding can be seen as a generalization of the situation where a good geometric dual exists, though the geometric intuition may still prove useful even in that case.

\subsection{Superpositions of Geometries at Large Times}
For times much larger than the scrambling time, the evolution of any chaotic system becomes random. In this regime, a few bits of information can still be transmitted by the coupling. But the interpretation of this signal is not that the particle goes through a semi-classical wormhole, even if the quantum system is in a parameter regime (e.g., large $N$ and strong coupling) where a clean semi-classical description is possible. The reason is the butterfly effect: at large times, a small perturbation (putting in the particle) can destroy any correlations between the two sides that would have existed without the perturbation.
The strength of the negative energy shockwave in the bulk is directly proportional to the amount of correlation between the two sides; at very large times, the correlation is simply too weak to shift the particle out of the horizon.
Nevertheless, there is another effect \cite{maldacena2017diving} involving the interference of two macroscopically different states (or bulk geometries) that allows for information transfer that we will now explain.

Consider the insertion of a message at time $-t$ on the left system using the unitary operator $U_L = e^{i\epsilon \phi_L} \approx 1 + i\epsilon \phi_L$.\footnote{The small-$\epsilon$ approximation is not necessary for the conclusions of this section. See~\cite{gao2018regenesis}.}  At time $t=0$ we let the left and right systems interact, so that the state is $\ket \Phi= e^{igV}U_L(-t)\ket{\TFD}$. We know that the action of $e^{igV}$ depends on the size of the state on which it acts. The key fact is that the operator $\phi_L(-t)$, for large $t$, is a totally random operator. Therefore, its size is equal to that of a random operator, which is nearly maximal. So $e^{igV}$ acts simply as a relative phase $\theta \sim 1$ between $\ket{\TFD}$ and $\phi_L(-t)\ket{\TFD}$. We can think of it as a phase-shift gate. Then $\ket \Phi = \ket{\TFD} + i\epsilon e^{ig\theta}\phi_L(-t)\ket{\TFD}$. This state is a superposition of two vastly different geometries: one is an empty wormhole, given by the state $\ket{\TFD}$, while $\phi_L(-t)\ket{\TFD}$ contains an energetic particle with a significant backreaction on the geometry.

A simple way to record the receipt of the message is to compute the change of the expectation value of $\phi_R(t)$:
\begin{equation}
\bra \Phi \phi_R(t)\ket \Phi - \langle \phi_R(t)\rangle_\textrm{therm} =  2\epsilon \sin (g \theta)\langle\phi_R\phi_L\rangle_\textrm{therm}.
\end{equation}
See~\cite{gao2018regenesis,maldacena2017diving} for similar calculations.
Clearly, this scenario does not have the interpretation of a classically traversable wormhole. In fact, there is not much geometry left in the description at all. This scenario is contrasted with the situation at shorter times, where we have access to multiple eigenvalues of $e^{igV}$ and the momentum-size correspondence has a clear geometric meaning. In all cases, the dynamics of the phase in the size distribution gives the right description of the physics, but there is a transition from a classical to a quantum picture.

\subsection{Wormhole Tomography and Other Future Directions}
There are a number of interesting future directions for investigation. We have focused on two regimes, one relatively short (slightly before the scrambling time) where the particle classical traverses, and the long time effect, which involves interference. This of course does not exhaust the list of non-geometric effects; for example, stringy effects can play an important role at finite coupling, when the string scale is not parametrically suppressed \cite{maldacena2017diving}. We have started to explore this in the analytically-tractable playground of the large-$q$ SYK model at finite $\beta \mathcal{J}$~\cite{longpaper}.

One might wonder whether it is really possible to operationally distinguish whether the information went ``through'' the wormhole, or not. We propose the following criteria: if the black hole is in a state where there is a diary behind both horizons, a protocol which involves teleportation ``through'' the wormhole should be sensitive to what is in the diary. In other words, if Bob claims that he went through a wormhole to get to Alice, we can ask him to prove it by giving some description of what was inside the black hole. If we send multiple observers through, they should share information about the interior that is consistent with each other.

In the classically-traversable case, one can imagine therefore engaging in ``wormhole tomography,'' where the contents of the wormhole interior (as determined by some non-TFD initial state) are probed experimentally by state transfer experiments; the signal exiting the wormhole will be modified in some way by the particular geometry of the wormhole and the presence or absence of any matter.

We analyze size winding in the SYK model in~\cite{longpaper}, but there are still some open questions about the details of the state transfer protocol in the case where it corresponds to a through-the-wormhole process. Rather than simply swapping a physical qubit with the message qubit, as we have advocated here, one wants to swap the message qubit with a logical qubit that represents, say, the polarization states of an emergent bulk photon. The key fact about this distinction is that the logical subspace for the encoding has fixed bulk energy, so the gravitational backreaction does not depend on the message. This is one way to avoid superpositions of macroscopically different geometries. There is no obstruction preventing us from carrying out this task in principle, and it might be instructive to actually do it. The problem is one of engineering, and a more complicated model like $\mathcal{N}=4$ super Yang-Mills theory might be required in order to have the necessary ingredients.

\section{Experimental Realization}\label{sec:experiment}
As discussed above, this work concerns a whole family of protocols, all of which are interesting to study experimentally for the light they would shed on entanglement, chaos, and holography. For example, if the system under study has a simple dual holographic description, such as the SYK model~\cite{Sachdev1993Gapless,Georges2000Mean,Georges2001Quantum,kitaev,Polchinski_2016,ms,kamenev,Garc_a_Garc_a_2016} or certain supersymmetric gauge theories, the experiments described here can directly probe traversable wormholes. More generally, these experiments probe communication phenomena inspired by and related to the traversable wormhole phenomenon in holographic models. The key ingredients are as follows.

First, one must be able to prepare a thermofield double state associated with $H$. This means preparing a special entangled state of two copies of the physical system, the left and right systems. At infinite temperature, the thermofield double state is just a collection of Bell pairs between left and right (or the appropriate fermionic version). For general Hamiltonians and non-infinite temperature, there is no known procedure to prepare the thermofield double state. However, there are recently proposed approximate methods that are applicable to systems of interest including the SYK model and various spin chains~\cite{Martyn_2019,Cottrell_2019,wu2018variational}.

Second, one must be able to effectively evolve forward and backward in time with the system Hamiltonian $H$. More precisely, we require the ability to evolve forward and backward with $H_L = H$ on the left system and the ability to evolve forward with the CPT conjugate of $H$, $H_R = H^T$, on the right system. Given a fully controlled fault-tolerant quantum computer and a Trotterized approximation of $e^{-i H t}$, it is in principle no more challenging to implement $e^{+ i H t}$ (backward evolution) than it is to implement $e^{-i H t}$ (forward evolution). However, implementing forward and backward time evolution in a specialized quantum simulator requires specific capabilities. In the context of measurements of out-of-time-order correlators, various techniques have been developed to achieve this level of control, at least approximately~\cite{Swingle:2016var,Yao2016a,Zhu2016,Campisi2017,Halpern2016,Yoshida2018,Li2017a,Garttner2016,Meier2017,Landsman_2019}.

Third, one must be able to apply the weak left-right coupling given by the $V$ operator. More precisely, it must be possible to generate the unitary $e^{i g V}$. This coupling must be applied suddenly, in between the other time-evolution segments of the circuit.

Fourth, one must be able to apply local control operations, including deleting and inserting qubits, performing local unitary operations, and making local measurements in a general basis. This requires some degree of individual qubit addressability, although in the simplest cases one only needs to single out a small number of qubits.

Given these capabilities, the general protocols in \cref{fig:Wormhole_Circuit} can be carried out. For concreteness, the remainder of this section will focus on the case of the insertion/deletion protocol, \cref{fig:Wormhole_Circuit}(a).
To give an example, consider the deletion/insertion protocol at infinite temperature when $g=\pi$, all times involved are large, $n$ is very large, and $m=1$. In this case,  $\Psi_{\text{out}} = Y \Psi_{\text{in}} Y$ with perfect fidelity.

\subsection{Rydberg atom arrays}
One platform where such phenomena could be studied is Rydberg atom arrays. In one implementation~\cite{Bernien_2017}, information is encoded in a pair of levels in $^{87}$Rb, a ground state $|g\rangle$ and a Rydberg state $|r\rangle$, such that the effective Hamiltonian can be written in a spin-chain form as
\begin{equation*}
    H= \sum_i \frac{\Omega_i}{2}X_i + \sum_i \Delta_i \frac{I - Z_i}{2} + \frac{1}{4}\sum_{i<j}V_{ij}(I-Z_i)(I-Z_j),
\end{equation*}
where $Z_i = \proj{g_i} - \proj{r_i}$ and $X_i = \ketbra{g_i}{r_i} + \ketbra{r_i}{g_i}$, $\Omega_i$ and $\Delta_i$ are tunable field parameters, and $V_{ij}$ is the van der Waals interaction between the atoms.

In terms of capabilities listed above, preparation of an infinite-temperature thermofield double state (i.e., Bell pairs) has already been achieved using Rydberg atoms~\cite{RydbergBellPairs}. For finite temperatures, the approximate methods discussed above could also be applied to this setup. One can engineer the requisite backwards time evolution in various ways. One possibility is to work in the blockade regime, in which the effective dynamics takes place in a constrained Hilbert space and is governed just by the fields $\Omega$ and $\Delta$. These parameters can be reversed with echo pulse sequences and so forward and backward evolution is possible. Below we will also discuss a different Floquet scheme. The left/right coupling $V = \frac{1}{n-m} \sum_i Z_i^L Z_i^R$ is also feasible in a Rydberg system, and is already needed to prepare the Bell states. Finally, local addressing is possible and localized readout has been demonstrated~\cite{Bernien_2017}.

\begin{figure}
    \centering
    \includegraphics[width=\textwidth]{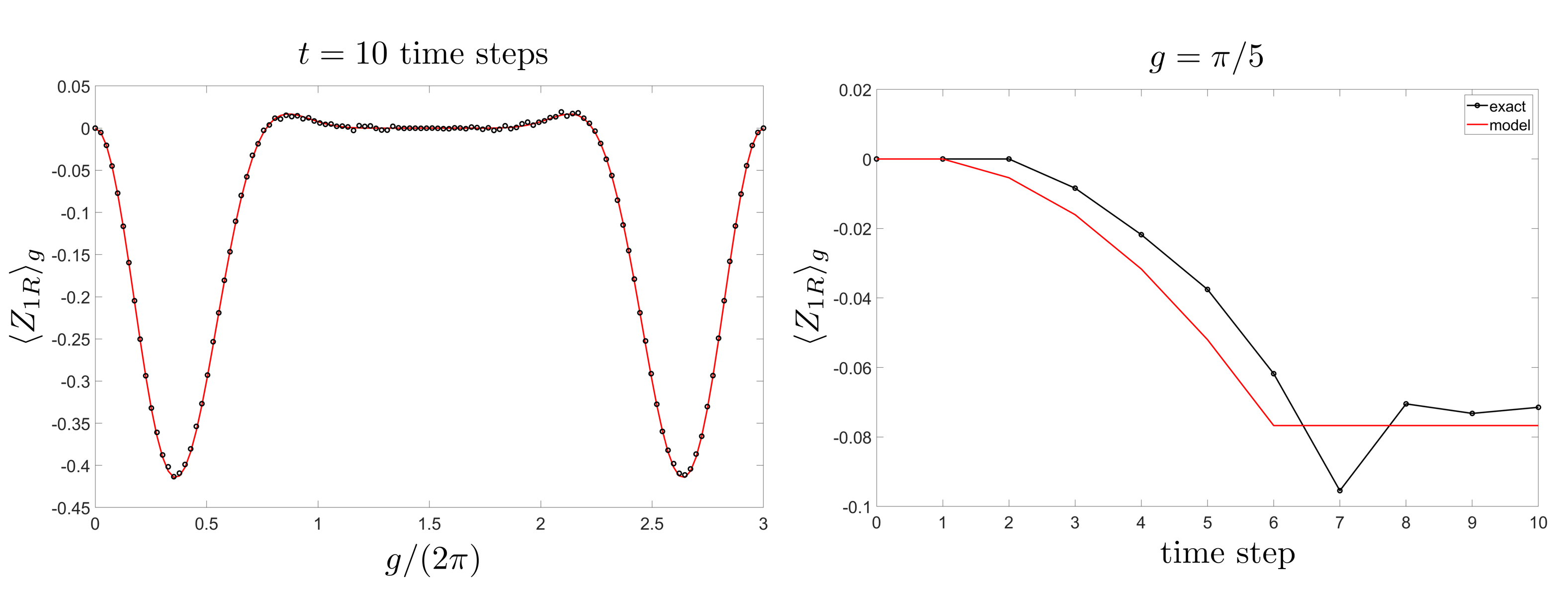}
    \caption{Expectation value of $Z_{1R}$ after injection of $Z_{1L}=1$ state on the left system. Black dots are direct numerical simulation of the protocol in the quantum kicked Ising model with $n=7$ spins on the left and right and with $J=b=\pi/4$ and $h_i$ drawn from a box distribution of width $.5$. \textbf{Left:} Signal at fixed large time as a function of $g$. The black circles are the exact numerical simulation. The red curve is the theory prediction in Eq.~\eqref{eq:zavg_infT}. \textbf{Right:} Signal at fixed $g$ as a function of time step. The black circles are the exact numerical simulation. The red curve is a crude approximation where we assume Eq.~\eqref{eq:zavg_infT} holds at all times with the effective system size replaced as $n \rightarrow \min(t+1,n)$.   }
    \label{fig:prosen_lab}
\end{figure}

One particularly interesting system to consider is a Floquet version of the Rydberg Hamiltonian known as the kicked quantum Ising model. Although experiments here are naturally restricted to infinite temperature because of heating, the driving is interesting because it can enhance chaos and aid in the problem of backward and forward time evolution. Consider, for example, the kicked quantum Ising model of Prosen et al.~\cite{Bertini_2019}, in which the time evolution for one time step is given by
\begin{equation*}
    U = U_K U_I
\end{equation*}
where
\begin{equation*}
    U_K = \exp\left( i b \sum_i X_i \right)
\end{equation*}
and
\begin{equation*}
    U_I = \exp \left( i J \sum_i Z_i Z_{i+1} + i \sum_i h_i Z_i \right).
\end{equation*}
The parameters of the model are $J$, $b$, and the set of local fields $h_i$. Remarkably, if $J=b=\pi/4$ and $h_i$ are drawn uniformly at random from a Gaussian distribution with variance $\sigma$, then the model is in a sense maximally chaotic (albeit not in the out-of-time-order correlator sense). For example, the entanglement entropy of subsystems grows as rapidly as possible when starting from a product state~\cite{Bertini_2019}. We note that a hyperfine encoding for qubits (instead of directly using the Rydberg level) might be useful for this kind of gate-like time dynamics~\cite{RydbergBellPairs}.

This kicked model is particularly appealing because the infinite-temperature thermofield double state is easier to prepare and because it allows easier control over the evolution. This relative ease is due to the fact that the spectra of $\sum_i X_i$ and $\sum_i Z_i Z_{i+1}$ are integer, so that one has, for example, $U_K(b+2\pi) = U_K(b)$. Thus, backward evolution corresponds to $U_K(b)^{-1} = U_K(-b) = U_K(2\pi-b)$, so one can achieve backward time evolution by over-evolving in the forward direction. This covers the transverse field and interaction terms; the longitudinal field terms can be dealt with using a standard echo sequence. One important point is that if the left evolution for one time step is $U$, then the right evolution for one time step must be $U^T = U_I U_K$ (note the reverse ordering of the pulsed terms, which are individually symmetric).

In \cref{fig:prosen_lab}, we show an exact numerical simulation of the experimental protocol for $n=7$ atoms on the left and right. We inject a pure state with eigenvalue $Z_{1L}=1$ into the first qubit on the left. Then, as a more experimentally accessible stand-in for the full fidelity, we show the result of measuring the expectation value of $Z_{1R}$ on the right. The black dots are the exact simulation and the red curves are obtained from our theory calculations. In particular, for a system with $n$ atoms and left-right coupling $g$ at large time, the prediction for the expectation value is
\begin{equation}\label{eq:zavg_infT}
    \langle Z_{1R}\rangle_g =  \left(\cos \frac{g}{n-1} \right)^{n-1} \frac{- \left(\cos \frac{g}{n-1} \right)^{n-1} + \cos g}{2}.
\end{equation}
As can be seen from the left panel of \cref{fig:prosen_lab}, the theory prediction perfectly fits the exact simulation data in the kicked quantum Ising model.

\subsection{Trapped ions}
While the Rydberg atom arrays just discussed have a natural spatial structure to their interactions, it is also quite interesting to consider systems which can support few-body but geometrically non-local interactions. One such system is an ion trap quantum processor, e.g.~\cite{wright2019benchmarking}, a version of which has already been used to study a wormhole-inspired protocol~\cite{Landsman_2019}. By driving vibrational modes of an ionic crystal, one can engineer a rich pattern of all-to-all interactions~\cite{davoudi2019analog}. Such systems are interesting because they mimic the structure of the SYK model and other matrix models that exhibit low-energy dynamics governed by a simple gravitational effective theory. One can again consider analog or digital versions of the platform, and in the digital case all the needed capabilities are present. Particularly interesting is a recent small-scale preparation of approximate thermofield double states on such a digital trapped ion quantum processor~\cite{zhu2019variational}.

\section{Closing remarks}
We have discussed two candidate systems, but many other platforms should be able to realize the physics discussed here. In our companion paper~\cite{longpaper}, we study a wide variety of models, including spin chains, random circuits, random Hamiltonians, and the SYK model, and some of these would be more naturally suited to other platforms, for example, proposals to realized SYK in simulators~\cite{Danshita_2017,Pikulin_2017} or on digital devices~\cite{Garc_a_lvarez_2017,Babbush_2019}.

In closing, let us highlight some of the conceptual and practical issues that will be faced in any experimental effort along the lines we discuss here. On the practical side, one key question is the impact of noise and experimental imperfections on our protocols, especially imperfect time-reversal due to over- or under-evolution and effects of environmental decoherence. Preliminary simulations indicate that the basic physics can still be seen when imperfections are below the 5$\%$ level for modest system size and time, but much more study is needed in the context of particular platforms. This general class of observables does exhibit some forms of resilience~\cite{Swingle_2018}. Another crucial question is how well the thermofield double state must be prepared to see the physics we discuss.

On the conceptual side, we must ask what we ultimately hope to learn about nature from such experiments. We emphasized above that the infinite-temperature large-time example does not correspond to geometrical motion through a semi-classical wormhole. For one thing, only a single qubit can be teleported with high fidelity in the high-temperature limit, but with the right encoding of information many qubits can be sent at low temperature and intermediate time in a holographic system hosting a traversable wormhole. Instead, the infinite-temperature example probes a physical effect common to all chaotic quantum systems with many degrees of freedom, including quantum gravitational systems.

From these considerations, it should be clear that measuring a successful teleportation signal for a single qubit is not enough to guarantee a semi-classical traversable wormhole in the bulk. One needs additional conditions that can be tested within the framework discussed here by varying the time $t$, the coupling $g$, and the way input information is encoded. Hence, while one long-term goal of such experiments is to detect and study wormholes arising holographically in highly entangled systems, there are other goals. More generally, the purpose is to shed light on deep and theoretically challenging questions about nature, including the necessary conditions to have a semi-classical bulk and the effects of quantum and stringy corrections to the semi-classical gravity picture. Thus, we believe the experiments described here are worth the effort to realize the long-term potential for experimental insight into quantum gravity.

\ssection{Acknowledgements}
We thank Patrick Hayden, Bryce Kobrin, Richard Kueng, Misha Lukin, Chris Monroe, Geoff Penington, John Preskill, Xiaoliang Qi, Thomas Schuster, Douglas Stanford, Alexandre Streicher, Zhenbin Yang, and Norman Yao for fruitful discussions.
We also thank Iris Cong, Emil Khabiboulline, Harry Levine, Misha Lukin, Hannes Pichler, and Cris Zanoci for collaboration on related work.
H.G.~is supported by the Simons Foundation through It from Qubit collaboration.
H.L.~is supported by an NDSEG fellowship.
G.S.~is supported by DOE award Quantum Error Correction and Spacetime Geometry DE-SC0018407, the Simons Foundation via It From Qubit, and the IQIM at Caltech (NSF Grant PHY-1733907).
L.S.~is supported by NSF Award Number 1316699.
B.S.~acknowledges support from the Simons Foundation via It From Qubit and from the Department of Energy via the GeoFlow consortium.
M.W.~is supported by an NWO Veni grant no.~680-47-459.
The work of G.S.~was performed before joining Amazon Web Services.

\bibliographystyle{apsrev4-2}
\bibliography{paper1}

\begin{thebibliography}{65}%
\makeatletter
\providecommand \@ifxundefined [1]{%
 \@ifx{#1\undefined}
}%
\providecommand \@ifnum [1]{%
 \ifnum #1\expandafter \@firstoftwo
 \else \expandafter \@secondoftwo
 \fi
}%
\providecommand \@ifx [1]{%
 \ifx #1\expandafter \@firstoftwo
 \else \expandafter \@secondoftwo
 \fi
}%
\providecommand \natexlab [1]{#1}%
\providecommand \enquote  [1]{``#1''}%
\providecommand \bibnamefont  [1]{#1}%
\providecommand \bibfnamefont [1]{#1}%
\providecommand \citenamefont [1]{#1}%
\providecommand \href@noop [0]{\@secondoftwo}%
\providecommand \href [0]{\begingroup \@sanitize@url \@href}%
\providecommand \@href[1]{\@@startlink{#1}\@@href}%
\providecommand \@@href[1]{\endgroup#1\@@endlink}%
\providecommand \@sanitize@url [0]{\catcode `\\12\catcode `\$12\catcode
  `\&12\catcode `\#12\catcode `\^12\catcode `\_12\catcode `\%12\relax}%
\providecommand \@@startlink[1]{}%
\providecommand \@@endlink[0]{}%
\providecommand \url  [0]{\begingroup\@sanitize@url \@url }%
\providecommand \@url [1]{\endgroup\@href {#1}{\urlprefix }}%
\providecommand \urlprefix  [0]{URL }%
\providecommand \Eprint [0]{\href }%
\providecommand \doibase [0]{https://doi.org/}%
\providecommand \selectlanguage [0]{\@gobble}%
\providecommand \bibinfo  [0]{\@secondoftwo}%
\providecommand \bibfield  [0]{\@secondoftwo}%
\providecommand \translation [1]{[#1]}%
\providecommand \BibitemOpen [0]{}%
\providecommand \bibitemStop [0]{}%
\providecommand \bibitemNoStop [0]{.\EOS\space}%
\providecommand \EOS [0]{\spacefactor3000\relax}%
\providecommand \BibitemShut  [1]{\csname bibitem#1\endcsname}%
\let\auto@bib@innerbib\@empty
\bibitem [{\citenamefont {Gao}\ \emph {et~al.}(2017)\citenamefont {Gao},
  \citenamefont {Jafferis},\ and\ \citenamefont {Wall}}]{gao2017traversable}%
  \BibitemOpen
  \bibfield  {author} {\bibinfo {author} {\bibfnamefont {P.}~\bibnamefont
  {Gao}}, \bibinfo {author} {\bibfnamefont {D.~L.}\ \bibnamefont {Jafferis}},\
  and\ \bibinfo {author} {\bibfnamefont {A.~C.}\ \bibnamefont {Wall}},\ }\href
  {https://doi.org/10.1007/JHEP12(2017)151} {\bibfield  {journal} {\bibinfo
  {journal} {J. High Energ. Phys.}\ }\textbf {\bibinfo {volume} {2017}}\bibinfo
   {number} { (12)},\ \bibinfo {pages} {151}}\BibitemShut {NoStop}%
\bibitem [{\citenamefont {Maldacena}\ \emph {et~al.}(2017)\citenamefont
  {Maldacena}, \citenamefont {Stanford},\ and\ \citenamefont
  {Yang}}]{maldacena2017diving}%
  \BibitemOpen
\bibfield  {number} {  }\bibfield  {author} {\bibinfo {author} {\bibfnamefont
  {J.}~\bibnamefont {Maldacena}}, \bibinfo {author} {\bibfnamefont
  {D.}~\bibnamefont {Stanford}},\ and\ \bibinfo {author} {\bibfnamefont
  {Z.}~\bibnamefont {Yang}},\ }\href {https://doi.org/10.1002/prop.201700034}
  {\bibfield  {journal} {\bibinfo  {journal} {Fortschritte der Phys.}\ }\textbf
  {\bibinfo {volume} {65}},\ \bibinfo {pages} {1700034} (\bibinfo {year}
  {2017})}\BibitemShut {NoStop}%
\bibitem [{\citenamefont {Nezami}\ \emph {et~al.}(2021)\citenamefont {Nezami},
  \citenamefont {Lin}, \citenamefont {Brown}, \citenamefont {Gharibyan},
  \citenamefont {Leichenauer}, \citenamefont {Salton}, \citenamefont
  {Susskind}, \citenamefont {Swingle},\ and\ \citenamefont
  {Walter}}]{longpaper}%
  \BibitemOpen
  \bibfield  {author} {\bibinfo {author} {\bibfnamefont {S.}~\bibnamefont
  {Nezami}}, \bibinfo {author} {\bibfnamefont {H.~W.}\ \bibnamefont {Lin}},
  \bibinfo {author} {\bibfnamefont {A.~R.}\ \bibnamefont {Brown}}, \bibinfo
  {author} {\bibfnamefont {H.}~\bibnamefont {Gharibyan}}, \bibinfo {author}
  {\bibfnamefont {S.}~\bibnamefont {Leichenauer}}, \bibinfo {author}
  {\bibfnamefont {G.}~\bibnamefont {Salton}}, \bibinfo {author} {\bibfnamefont
  {L.}~\bibnamefont {Susskind}}, \bibinfo {author} {\bibfnamefont
  {B.}~\bibnamefont {Swingle}},\ and\ \bibinfo {author} {\bibfnamefont
  {M.}~\bibnamefont {Walter}},\ }\href@noop {} {\bibinfo {title} {{Q}uantum
  {G}ravity in the {L}ab: {T}eleportation by {S}ize and {T}raversable
  {W}ormholes, {P}art {II}}} (\bibinfo {year} {2021}),\ \bibinfo {note} {arXiv
  preprint, Jan 2021}\BibitemShut {NoStop}%
\bibitem [{\citenamefont {'t~Hooft}(1993)}]{tHooft:1993dmi}%
  \BibitemOpen
  \bibfield  {author} {\bibinfo {author} {\bibfnamefont {G.}~\bibnamefont
  {'t~Hooft}},\ }\bibfield  {booktitle} {\emph {\bibinfo {booktitle}
  {{Conference on Highlights of Particle and Condensed Matter Physics
  (SALAMFEST) Trieste, Italy, March 8-12, 1993}}},\ }\href@noop {} {\bibfield
  {journal} {\bibinfo  {journal} {Conf. Proc.}\ }\textbf {\bibinfo {volume}
  {C930308}},\ \bibinfo {pages} {284} (\bibinfo {year} {1993})},\ \Eprint
  {https://arxiv.org/abs/gr-qc/9310026} {arXiv:gr-qc/9310026} \BibitemShut
  {NoStop}%
\bibitem [{\citenamefont {Susskind}(1995)}]{Susskind:1994vu}%
  \BibitemOpen
  \bibfield  {author} {\bibinfo {author} {\bibfnamefont {L.}~\bibnamefont
  {Susskind}},\ }\href {https://doi.org/10.1063/1.531249} {\bibfield  {journal}
  {\bibinfo  {journal} {J. Math. Phys.}\ }\textbf {\bibinfo {volume} {36}},\
  \bibinfo {pages} {6377} (\bibinfo {year} {1995})}\BibitemShut {NoStop}%
\bibitem [{\citenamefont {Maldacena}(1999)}]{Maldacena:1997re}%
  \BibitemOpen
  \bibfield  {author} {\bibinfo {author} {\bibfnamefont {J.~M.}\ \bibnamefont
  {Maldacena}},\ }\href {https://doi.org/10.1023/A:1026654312961,
  10.4310/ATMP.1998.v2.n2.a1} {\bibfield  {journal} {\bibinfo  {journal} {Int.
  J. Theor. Phys.}\ }\textbf {\bibinfo {volume} {38}},\ \bibinfo {pages} {1113}
  (\bibinfo {year} {1999})}\BibitemShut {NoStop}%
\bibitem [{\citenamefont {Gubser}\ \emph {et~al.}(1998)\citenamefont {Gubser},
  \citenamefont {Klebanov},\ and\ \citenamefont {Polyakov}}]{Gubser_1998}%
  \BibitemOpen
  \bibfield  {author} {\bibinfo {author} {\bibfnamefont {S.}~\bibnamefont
  {Gubser}}, \bibinfo {author} {\bibfnamefont {I.}~\bibnamefont {Klebanov}},\
  and\ \bibinfo {author} {\bibfnamefont {A.}~\bibnamefont {Polyakov}},\ }\href
  {https://doi.org/10.1016/s0370-2693(98)00377-3} {\bibfield  {journal}
  {\bibinfo  {journal} {Phys. Lett. B}\ }\textbf {\bibinfo {volume} {428}},\
  \bibinfo {pages} {105} (\bibinfo {year} {1998})}\BibitemShut {NoStop}%
\bibitem [{\citenamefont {Witten}(1998)}]{witten1998anti}%
  \BibitemOpen
  \bibfield  {author} {\bibinfo {author} {\bibfnamefont {E.}~\bibnamefont
  {Witten}},\ }\href {https://doi.org/10.4310/ATMP.1998.v2.n2.a2} {\bibfield
  {journal} {\bibinfo  {journal} {Adv. Theor. Math. Phys.}\ }\textbf {\bibinfo
  {volume} {2}},\ \bibinfo {pages} {253} (\bibinfo {year} {1998})}\BibitemShut
  {NoStop}%
\bibitem [{\citenamefont {Banks}\ \emph {et~al.}(1997)\citenamefont {Banks},
  \citenamefont {Fischler}, \citenamefont {Shenker},\ and\ \citenamefont
  {Susskind}}]{Banks:1996vh}%
  \BibitemOpen
  \bibfield  {author} {\bibinfo {author} {\bibfnamefont {T.}~\bibnamefont
  {Banks}}, \bibinfo {author} {\bibfnamefont {W.}~\bibnamefont {Fischler}},
  \bibinfo {author} {\bibfnamefont {S.~H.}\ \bibnamefont {Shenker}},\ and\
  \bibinfo {author} {\bibfnamefont {L.}~\bibnamefont {Susskind}},\ }\href
  {https://doi.org/10.1103/PhysRevD.55.5112} {\bibfield  {journal} {\bibinfo
  {journal} {Phys. Rev.}\ }\textbf {\bibinfo {volume} {D55}},\ \bibinfo {pages}
  {5112} (\bibinfo {year} {1997})}\BibitemShut {NoStop}%
\bibitem [{\citenamefont {Sekino}\ and\ \citenamefont
  {Susskind}(2008)}]{Sekino_2008}%
  \BibitemOpen
  \bibfield  {author} {\bibinfo {author} {\bibfnamefont {Y.}~\bibnamefont
  {Sekino}}\ and\ \bibinfo {author} {\bibfnamefont {L.}~\bibnamefont
  {Susskind}},\ }\href {https://doi.org/10.1088/1126-6708/2008/10/065}
  {\bibfield  {journal} {\bibinfo  {journal} {J. High Energ. Phys.}\ }\textbf
  {\bibinfo {volume} {2008}}\bibinfo  {number} { (10)},\ \bibinfo {pages}
  {065}}\BibitemShut {NoStop}%
\bibitem [{\citenamefont {Hayden}\ and\ \citenamefont
  {Preskill}(2007)}]{Hayden:2007cs}%
  \BibitemOpen
\bibfield  {number} {  }\bibfield  {author} {\bibinfo {author} {\bibfnamefont
  {P.}~\bibnamefont {Hayden}}\ and\ \bibinfo {author} {\bibfnamefont
  {J.}~\bibnamefont {Preskill}},\ }\href
  {https://doi.org/10.1088/1126-6708/2007/09/120} {\bibfield  {journal}
  {\bibinfo  {journal} {J. High Energ. Phys.}\ }\textbf {\bibinfo {volume}
  {2007}}\bibinfo  {number} { (09)},\ \bibinfo {pages} {120}}\BibitemShut
  {NoStop}%
\bibitem [{\citenamefont {{Bennett}}\ \emph {et~al.}(1993)\citenamefont
  {{Bennett}}, \citenamefont {{Brassard}}, \citenamefont {{Crepeau}},
  \citenamefont {{Jozsa}}, \citenamefont {{Peres}},\ and\ \citenamefont
  {{Wootters}}}]{QTeleportation}%
  \BibitemOpen
\bibfield  {number} {  }\bibfield  {author} {\bibinfo {author} {\bibfnamefont
  {C.~H.}\ \bibnamefont {{Bennett}}}, \bibinfo {author} {\bibfnamefont
  {G.}~\bibnamefont {{Brassard}}}, \bibinfo {author} {\bibfnamefont
  {C.}~\bibnamefont {{Crepeau}}}, \bibinfo {author} {\bibfnamefont
  {R.}~\bibnamefont {{Jozsa}}}, \bibinfo {author} {\bibfnamefont
  {A.}~\bibnamefont {{Peres}}},\ and\ \bibinfo {author} {\bibfnamefont {W.~K.}\
  \bibnamefont {{Wootters}}},\ }\href
  {https://doi.org/10.1103/PhysRevLett.70.1895} {\bibfield  {journal} {\bibinfo
   {journal} {\prl}\ }\textbf {\bibinfo {volume} {70}},\ \bibinfo {pages}
  {1895} (\bibinfo {year} {1993})}\BibitemShut {NoStop}%
\bibitem [{\citenamefont {Freivogel}\ \emph {et~al.}(2020)\citenamefont
  {Freivogel}, \citenamefont {Galante}, \citenamefont {Nikolakopoulou},\ and\
  \citenamefont {Rotundo}}]{freivogel2019traversable}%
  \BibitemOpen
  \bibfield  {author} {\bibinfo {author} {\bibfnamefont {B.}~\bibnamefont
  {Freivogel}}, \bibinfo {author} {\bibfnamefont {D.~A.}\ \bibnamefont
  {Galante}}, \bibinfo {author} {\bibfnamefont {D.}~\bibnamefont
  {Nikolakopoulou}},\ and\ \bibinfo {author} {\bibfnamefont {A.}~\bibnamefont
  {Rotundo}},\ }\href {https://doi.org/10.1007/JHEP01(2020)050} {\bibfield
  {journal} {\bibinfo  {journal} {J. High Energ. Phys.}\ }\textbf {\bibinfo
  {volume} {2020}}\bibinfo  {number} { (1)},\ \bibinfo {pages} {1}}\BibitemShut
  {NoStop}%
\bibitem [{\citenamefont {Bak}\ \emph {et~al.}(2019)\citenamefont {Bak},
  \citenamefont {Kim},\ and\ \citenamefont {Yi}}]{bak2019experimental}%
  \BibitemOpen
\bibfield  {number} {  }\bibfield  {author} {\bibinfo {author} {\bibfnamefont
  {D.}~\bibnamefont {Bak}}, \bibinfo {author} {\bibfnamefont {C.}~\bibnamefont
  {Kim}},\ and\ \bibinfo {author} {\bibfnamefont {S.-H.}\ \bibnamefont {Yi}},\
  }\href {https://doi.org/10.1007/JHEP12(2019)005} {\bibfield  {journal}
  {\bibinfo  {journal} {J. High Energ. Phys.}\ }\textbf {\bibinfo {volume}
  {2019}}\bibinfo  {number} { (12)},\ \bibinfo {pages} {1}}\BibitemShut
  {NoStop}%
\bibitem [{\citenamefont {Bao}\ \emph {et~al.}(2018)\citenamefont {Bao},
  \citenamefont {Chatwin-Davies}, \citenamefont {Pollack},\ and\ \citenamefont
  {Remmen}}]{Bao_2018}%
  \BibitemOpen
\bibfield  {number} {  }\bibfield  {author} {\bibinfo {author} {\bibfnamefont
  {N.}~\bibnamefont {Bao}}, \bibinfo {author} {\bibfnamefont {A.}~\bibnamefont
  {Chatwin-Davies}}, \bibinfo {author} {\bibfnamefont {J.}~\bibnamefont
  {Pollack}},\ and\ \bibinfo {author} {\bibfnamefont {G.~N.}\ \bibnamefont
  {Remmen}},\ }\href {https://doi.org/10.1007/jhep11(2018)071} {\bibfield
  {journal} {\bibinfo  {journal} {J. High Energ. Phys.}\ }\textbf {\bibinfo
  {volume} {2018}}\bibinfo  {number} { (11)}}\BibitemShut {NoStop}%
\bibitem [{\citenamefont {Bao}\ \emph {et~al.}(2019)\citenamefont {Bao},
  \citenamefont {Su},\ and\ \citenamefont {Usatyuk}}]{bao2019wormhole}%
  \BibitemOpen
\bibfield  {number} {  }\bibfield  {author} {\bibinfo {author} {\bibfnamefont
  {N.}~\bibnamefont {Bao}}, \bibinfo {author} {\bibfnamefont {V.~P.}\
  \bibnamefont {Su}},\ and\ \bibinfo {author} {\bibfnamefont {M.}~\bibnamefont
  {Usatyuk}},\ }\href@noop {} {\bibinfo {title} {Wormhole traversability via
  quantum random walks}} (\bibinfo {year} {2019}),\ \bibinfo {note} {arXiv
  preprint},\ \Eprint {https://arxiv.org/abs/1906.01672} {arXiv:1906.01672}
  \BibitemShut {NoStop}%
\bibitem [{\citenamefont {Gao}\ and\ \citenamefont
  {Jafferis}(2019)}]{gao2019traversable}%
  \BibitemOpen
  \bibfield  {author} {\bibinfo {author} {\bibfnamefont {P.}~\bibnamefont
  {Gao}}\ and\ \bibinfo {author} {\bibfnamefont {D.~L.}\ \bibnamefont
  {Jafferis}},\ }\href@noop {} {\bibinfo {title} {A traversable wormhole
  teleportation protocol in the {SYK} model}} (\bibinfo {year} {2019}),\
  \bibinfo {note} {arXiv preprint},\ \Eprint {https://arxiv.org/abs/1911.07416}
  {arXiv:1911.07416} \BibitemShut {NoStop}%
\bibitem [{\citenamefont {Landsman}\ \emph {et~al.}(2019)\citenamefont
  {Landsman}, \citenamefont {Figgatt}, \citenamefont {Schuster}, \citenamefont
  {Linke}, \citenamefont {Yoshida}, \citenamefont {Yao},\ and\ \citenamefont
  {Monroe}}]{Landsman_2019}%
  \BibitemOpen
  \bibfield  {author} {\bibinfo {author} {\bibfnamefont {K.~A.}\ \bibnamefont
  {Landsman}}, \bibinfo {author} {\bibfnamefont {C.}~\bibnamefont {Figgatt}},
  \bibinfo {author} {\bibfnamefont {T.}~\bibnamefont {Schuster}}, \bibinfo
  {author} {\bibfnamefont {N.~M.}\ \bibnamefont {Linke}}, \bibinfo {author}
  {\bibfnamefont {B.}~\bibnamefont {Yoshida}}, \bibinfo {author} {\bibfnamefont
  {N.~Y.}\ \bibnamefont {Yao}},\ and\ \bibinfo {author} {\bibfnamefont
  {C.}~\bibnamefont {Monroe}},\ }\href
  {https://doi.org/10.1038/s41586-019-0952-6} {\bibfield  {journal} {\bibinfo
  {journal} {Nature}\ }\textbf {\bibinfo {volume} {567}},\ \bibinfo {pages}
  {61} (\bibinfo {year} {2019})}\BibitemShut {NoStop}%
\bibitem [{\citenamefont {Yoshida}\ and\ \citenamefont
  {Kitaev}(2017)}]{Yoshida2017}%
  \BibitemOpen
  \bibfield  {author} {\bibinfo {author} {\bibfnamefont {B.}~\bibnamefont
  {Yoshida}}\ and\ \bibinfo {author} {\bibfnamefont {A.}~\bibnamefont
  {Kitaev}},\ }\href@noop {} {\bibinfo {title} {Efficient decoding for the
  hayden-preskill protocol}} (\bibinfo {year} {2017}),\ \bibinfo {note} {arXiv
  preprint},\ \Eprint {https://arxiv.org/abs/1710.03363} {arXiv:1710.03363}
  \BibitemShut {NoStop}%
\bibitem [{\citenamefont {Yoshida}\ and\ \citenamefont
  {Yao}(2019)}]{Yoshida2018}%
  \BibitemOpen
  \bibfield  {author} {\bibinfo {author} {\bibfnamefont {B.}~\bibnamefont
  {Yoshida}}\ and\ \bibinfo {author} {\bibfnamefont {N.~Y.}\ \bibnamefont
  {Yao}},\ }\href {https://arxiv.org/pdf/1803.10772.pdf
  http://arxiv.org/abs/1803.10772
  https://link.aps.org/doi/10.1103/PhysRevX.9.011006} {\bibfield  {journal}
  {\bibinfo  {journal} {Phys. Rev. X}\ }\textbf {\bibinfo {volume} {9}},\
  \bibinfo {pages} {011006} (\bibinfo {year} {2019})}\BibitemShut {NoStop}%
\bibitem [{\citenamefont {Susskind}(2019)}]{susskind2019complexity}%
  \BibitemOpen
  \bibfield  {author} {\bibinfo {author} {\bibfnamefont {L.}~\bibnamefont
  {Susskind}},\ }\href@noop {} {\bibinfo {title} {{Complexity and Newton's
  Laws}}} (\bibinfo {year} {2019}),\ \bibinfo {note} {arXiv preprint},\ \Eprint
  {https://arxiv.org/abs/1904.12819} {arXiv:1904.12819} \BibitemShut {NoStop}%
\bibitem [{\citenamefont {Lin}\ \emph {et~al.}(2019)\citenamefont {Lin},
  \citenamefont {Maldacena},\ and\ \citenamefont {Zhao}}]{lin2019symmetries}%
  \BibitemOpen
  \bibfield  {author} {\bibinfo {author} {\bibfnamefont {H.~W.}\ \bibnamefont
  {Lin}}, \bibinfo {author} {\bibfnamefont {J.}~\bibnamefont {Maldacena}},\
  and\ \bibinfo {author} {\bibfnamefont {Y.}~\bibnamefont {Zhao}},\ }\href
  {https://doi.org/10.1007/JHEP08(2019)049} {\bibfield  {journal} {\bibinfo
  {journal} {J. High Energ. Phys.}\ }\textbf {\bibinfo {volume} {2019}}\bibinfo
   {number} { (8)},\ \bibinfo {pages} {1}}\BibitemShut {NoStop}%
\bibitem [{\citenamefont {Qi}\ and\ \citenamefont
  {Streicher}(2019)}]{qi2019quantum}%
  \BibitemOpen
\bibfield  {number} {  }\bibfield  {author} {\bibinfo {author} {\bibfnamefont
  {X.-L.}\ \bibnamefont {Qi}}\ and\ \bibinfo {author} {\bibfnamefont
  {A.}~\bibnamefont {Streicher}},\ }\href
  {https://doi.org/10.1007/JHEP08(2019)012} {\bibfield  {journal} {\bibinfo
  {journal} {J. High Energ. Phys.}\ }\textbf {\bibinfo {volume} {2019}}\bibinfo
   {number} { (8)},\ \bibinfo {pages} {12}}\BibitemShut {NoStop}%
\bibitem [{\citenamefont {Nahum}\ \emph {et~al.}(2018)\citenamefont {Nahum},
  \citenamefont {Vijay},\ and\ \citenamefont {Haah}}]{Nahum_2018}%
  \BibitemOpen
\bibfield  {number} {  }\bibfield  {author} {\bibinfo {author} {\bibfnamefont
  {A.}~\bibnamefont {Nahum}}, \bibinfo {author} {\bibfnamefont
  {S.}~\bibnamefont {Vijay}},\ and\ \bibinfo {author} {\bibfnamefont
  {J.}~\bibnamefont {Haah}},\ }\href
  {https://doi.org/10.1103/physrevx.8.021014} {\bibfield  {journal} {\bibinfo
  {journal} {Phys. Rev. X}\ }\textbf {\bibinfo {volume} {8}},\ \bibinfo {pages}
  {021014} (\bibinfo {year} {2018})}\BibitemShut {NoStop}%
\bibitem [{\citenamefont {Von~Keyserlingk}\ \emph {et~al.}(2018)\citenamefont
  {Von~Keyserlingk}, \citenamefont {Rakovszky}, \citenamefont {Pollmann},\ and\
  \citenamefont {Sondhi}}]{von_Keyserlingk_2018}%
  \BibitemOpen
  \bibfield  {author} {\bibinfo {author} {\bibfnamefont {C.}~\bibnamefont
  {Von~Keyserlingk}}, \bibinfo {author} {\bibfnamefont {T.}~\bibnamefont
  {Rakovszky}}, \bibinfo {author} {\bibfnamefont {F.}~\bibnamefont
  {Pollmann}},\ and\ \bibinfo {author} {\bibfnamefont {S.~L.}\ \bibnamefont
  {Sondhi}},\ }\href {https://doi.org/10.1103/physrevx.8.021013} {\bibfield
  {journal} {\bibinfo  {journal} {Phys. Rev. X}\ }\textbf {\bibinfo {volume}
  {8}},\ \bibinfo {pages} {021013} (\bibinfo {year} {2018})}\BibitemShut
  {NoStop}%
\bibitem [{\citenamefont {Xu}\ and\ \citenamefont {Swingle}(2019)}]{Xu_2019}%
  \BibitemOpen
  \bibfield  {author} {\bibinfo {author} {\bibfnamefont {S.}~\bibnamefont
  {Xu}}\ and\ \bibinfo {author} {\bibfnamefont {B.}~\bibnamefont {Swingle}},\
  }\href {https://doi.org/10.1103/physrevx.9.031048} {\bibfield  {journal}
  {\bibinfo  {journal} {Phys. Rev. X}\ }\textbf {\bibinfo {volume} {9}},\
  \bibinfo {pages} {031048} (\bibinfo {year} {2019})}\BibitemShut {NoStop}%
\bibitem [{\citenamefont {Schumacher}(1996)}]{schumacher1996sending}%
  \BibitemOpen
  \bibfield  {author} {\bibinfo {author} {\bibfnamefont {B.}~\bibnamefont
  {Schumacher}},\ }\href {https://doi.org/10.1103/PhysRevA.54.2614} {\bibfield
  {journal} {\bibinfo  {journal} {Phys. Rev. A}\ }\textbf {\bibinfo {volume}
  {54}},\ \bibinfo {pages} {2614} (\bibinfo {year} {1996})}\BibitemShut
  {NoStop}%
\bibitem [{\citenamefont {Roberts}\ \emph {et~al.}(2018)\citenamefont
  {Roberts}, \citenamefont {Stanford},\ and\ \citenamefont
  {Streicher}}]{Roberts:2018mnp}%
  \BibitemOpen
  \bibfield  {author} {\bibinfo {author} {\bibfnamefont {D.~A.}\ \bibnamefont
  {Roberts}}, \bibinfo {author} {\bibfnamefont {D.}~\bibnamefont {Stanford}},\
  and\ \bibinfo {author} {\bibfnamefont {A.}~\bibnamefont {Streicher}},\ }\href
  {https://doi.org/10.1007/JHEP06(2018)122} {\bibfield  {journal} {\bibinfo
  {journal} {J. High Energ. Phys.}\ }\textbf {\bibinfo {volume} {2018}}\bibinfo
   {number} { (6)},\ \bibinfo {pages} {122}}\BibitemShut {NoStop}%
\bibitem [{\citenamefont {Brown}\ \emph {et~al.}(2018)\citenamefont {Brown},
  \citenamefont {Gharibyan}, \citenamefont {Streicher}, \citenamefont
  {Susskind}, \citenamefont {Thorlacius},\ and\ \citenamefont
  {Zhao}}]{Brown:2018kvn}%
  \BibitemOpen
\bibfield  {number} {  }\bibfield  {author} {\bibinfo {author} {\bibfnamefont
  {A.~R.}\ \bibnamefont {Brown}}, \bibinfo {author} {\bibfnamefont
  {H.}~\bibnamefont {Gharibyan}}, \bibinfo {author} {\bibfnamefont
  {A.}~\bibnamefont {Streicher}}, \bibinfo {author} {\bibfnamefont
  {L.}~\bibnamefont {Susskind}}, \bibinfo {author} {\bibfnamefont
  {L.}~\bibnamefont {Thorlacius}},\ and\ \bibinfo {author} {\bibfnamefont
  {Y.}~\bibnamefont {Zhao}},\ }\href
  {https://doi.org/10.1103/PhysRevD.98.126016} {\bibfield  {journal} {\bibinfo
  {journal} {Phys. Rev.}\ }\textbf {\bibinfo {volume} {D98}},\ \bibinfo {pages}
  {126016} (\bibinfo {year} {2018})}\BibitemShut {NoStop}%
\bibitem [{\citenamefont {Maldacena}\ and\ \citenamefont
  {Qi}(2018)}]{maldacena2018eternal}%
  \BibitemOpen
  \bibfield  {author} {\bibinfo {author} {\bibfnamefont {J.}~\bibnamefont
  {Maldacena}}\ and\ \bibinfo {author} {\bibfnamefont {X.-L.}\ \bibnamefont
  {Qi}},\ }\href@noop {} {\bibinfo {title} {Eternal traversable wormhole}}
  (\bibinfo {year} {2018}),\ \bibinfo {note} {arXiv preprint},\ \Eprint
  {https://arxiv.org/abs/1804.00491} {arXiv:1804.00491} \BibitemShut {NoStop}%
\bibitem [{\citenamefont {Gao}\ and\ \citenamefont
  {Liu}(2019)}]{gao2018regenesis}%
  \BibitemOpen
  \bibfield  {author} {\bibinfo {author} {\bibfnamefont {P.}~\bibnamefont
  {Gao}}\ and\ \bibinfo {author} {\bibfnamefont {H.}~\bibnamefont {Liu}},\
  }\href {https://doi.org/10.1007/JHEP10(2019)048} {\bibfield  {journal}
  {\bibinfo  {journal} {J. High Energ. Phys.}\ }\textbf {\bibinfo {volume}
  {2019}}\bibinfo  {number} { (10)},\ \bibinfo {pages} {1}}\BibitemShut
  {NoStop}%
\bibitem [{\citenamefont {Sachdev}\ and\ \citenamefont
  {Ye}(1993)}]{Sachdev1993Gapless}%
  \BibitemOpen
\bibfield  {number} {  }\bibfield  {author} {\bibinfo {author} {\bibfnamefont
  {S.}~\bibnamefont {Sachdev}}\ and\ \bibinfo {author} {\bibfnamefont
  {J.}~\bibnamefont {Ye}},\ }\href
  {https://doi.org/10.1103/PhysRevLett.70.3339} {\bibfield  {journal} {\bibinfo
   {journal} {Phys. Rev. Lett.}\ }\textbf {\bibinfo {volume} {70}},\ \bibinfo
  {pages} {3339} (\bibinfo {year} {1993})}\BibitemShut {NoStop}%
\bibitem [{\citenamefont {Georges}\ \emph {et~al.}(2000)\citenamefont
  {Georges}, \citenamefont {Parcollet},\ and\ \citenamefont
  {Sachdev}}]{Georges2000Mean}%
  \BibitemOpen
  \bibfield  {author} {\bibinfo {author} {\bibfnamefont {A.}~\bibnamefont
  {Georges}}, \bibinfo {author} {\bibfnamefont {O.}~\bibnamefont {Parcollet}},\
  and\ \bibinfo {author} {\bibfnamefont {S.}~\bibnamefont {Sachdev}},\ }\href
  {https://doi.org/10.1103/PhysRevLett.85.840} {\bibfield  {journal} {\bibinfo
  {journal} {Phys. Rev. Lett.}\ }\textbf {\bibinfo {volume} {85}},\ \bibinfo
  {pages} {840} (\bibinfo {year} {2000})}\BibitemShut {NoStop}%
\bibitem [{\citenamefont {Georges}\ \emph {et~al.}(2001)\citenamefont
  {Georges}, \citenamefont {Parcollet},\ and\ \citenamefont
  {Sachdev}}]{Georges2001Quantum}%
  \BibitemOpen
  \bibfield  {author} {\bibinfo {author} {\bibfnamefont {A.}~\bibnamefont
  {Georges}}, \bibinfo {author} {\bibfnamefont {O.}~\bibnamefont {Parcollet}},\
  and\ \bibinfo {author} {\bibfnamefont {S.}~\bibnamefont {Sachdev}},\ }\href
  {https://doi.org/10.1103/PhysRevB.63.134406} {\bibfield  {journal} {\bibinfo
  {journal} {Phys. Rev. B}\ }\textbf {\bibinfo {volume} {63}},\ \bibinfo
  {pages} {134406} (\bibinfo {year} {2001})}\BibitemShut {NoStop}%
\bibitem [{\citenamefont {Kitaev}(2015)}]{kitaev}%
  \BibitemOpen
  \bibfield  {author} {\bibinfo {author} {\bibfnamefont {A.}~\bibnamefont
  {Kitaev}},\ }\href@noop {} {\bibinfo {title} {A simple model of quantum
  holography}} (\bibinfo {year} {2015}),\ \bibinfo {note} {{T}alks at KITP,
  April 7, 2015 and May 27, 2015}\BibitemShut {NoStop}%
\bibitem [{\citenamefont {Polchinski}\ and\ \citenamefont
  {Rosenhaus}(2016)}]{Polchinski_2016}%
  \BibitemOpen
  \bibfield  {author} {\bibinfo {author} {\bibfnamefont {J.}~\bibnamefont
  {Polchinski}}\ and\ \bibinfo {author} {\bibfnamefont {V.}~\bibnamefont
  {Rosenhaus}},\ }\href {https://doi.org/10.1007/jhep04(2016)001} {\bibfield
  {journal} {\bibinfo  {journal} {J. High Energ. Phys.}\ }\textbf {\bibinfo
  {volume} {2016}}\bibinfo  {number} { (4)},\ \bibinfo {pages} {1}}\BibitemShut
  {NoStop}%
\bibitem [{\citenamefont {Maldacena}\ and\ \citenamefont
  {Stanford}(2016)}]{ms}%
  \BibitemOpen
\bibfield  {number} {  }\bibfield  {author} {\bibinfo {author} {\bibfnamefont
  {J.}~\bibnamefont {Maldacena}}\ and\ \bibinfo {author} {\bibfnamefont
  {D.}~\bibnamefont {Stanford}},\ }\href
  {https://doi.org/10.1103/PhysRevD.94.106002} {\bibfield  {journal} {\bibinfo
  {journal} {Phys. Rev.}\ }\textbf {\bibinfo {volume} {D94}},\ \bibinfo {pages}
  {106002} (\bibinfo {year} {2016})}\BibitemShut {NoStop}%
\bibitem [{\citenamefont {Bagrets}\ \emph {et~al.}(2017)\citenamefont
  {Bagrets}, \citenamefont {Altland},\ and\ \citenamefont {Kamenev}}]{kamenev}%
  \BibitemOpen
  \bibfield  {author} {\bibinfo {author} {\bibfnamefont {D.}~\bibnamefont
  {Bagrets}}, \bibinfo {author} {\bibfnamefont {A.}~\bibnamefont {Altland}},\
  and\ \bibinfo {author} {\bibfnamefont {A.}~\bibnamefont {Kamenev}},\ }\href
  {https://doi.org/10.1016/j.nuclphysb.2017.06.012} {\bibfield  {journal}
  {\bibinfo  {journal} {Nucl. Phys.}\ }\textbf {\bibinfo {volume} {B921}},\
  \bibinfo {pages} {727} (\bibinfo {year} {2017})}\BibitemShut {NoStop}%
\bibitem [{\citenamefont {Garc{\'\i}a-Garc{\'\i}a}\ and\ \citenamefont
  {Verbaarschot}(2016)}]{Garc_a_Garc_a_2016}%
  \BibitemOpen
  \bibfield  {author} {\bibinfo {author} {\bibfnamefont {A.~M.}\ \bibnamefont
  {Garc{\'\i}a-Garc{\'\i}a}}\ and\ \bibinfo {author} {\bibfnamefont {J.~J.}\
  \bibnamefont {Verbaarschot}},\ }\href
  {https://doi.org/10.1103/physrevd.94.126010} {\bibfield  {journal} {\bibinfo
  {journal} {Phys. Rev. D}\ }\textbf {\bibinfo {volume} {94}},\ \bibinfo
  {pages} {126010} (\bibinfo {year} {2016})}\BibitemShut {NoStop}%
\bibitem [{\citenamefont {Martyn}\ and\ \citenamefont
  {Swingle}(2019)}]{Martyn_2019}%
  \BibitemOpen
  \bibfield  {author} {\bibinfo {author} {\bibfnamefont {J.}~\bibnamefont
  {Martyn}}\ and\ \bibinfo {author} {\bibfnamefont {B.}~\bibnamefont
  {Swingle}},\ }\href {https://doi.org/10.1103/physreva.100.032107} {\bibfield
  {journal} {\bibinfo  {journal} {Phys. Rev. A}\ }\textbf {\bibinfo {volume}
  {100}},\ \bibinfo {pages} {032107} (\bibinfo {year} {2019})}\BibitemShut
  {NoStop}%
\bibitem [{\citenamefont {Cottrell}\ \emph {et~al.}(2019)\citenamefont
  {Cottrell}, \citenamefont {Freivogel}, \citenamefont {Hofman},\ and\
  \citenamefont {Lokhande}}]{Cottrell_2019}%
  \BibitemOpen
  \bibfield  {author} {\bibinfo {author} {\bibfnamefont {W.}~\bibnamefont
  {Cottrell}}, \bibinfo {author} {\bibfnamefont {B.}~\bibnamefont {Freivogel}},
  \bibinfo {author} {\bibfnamefont {D.~M.}\ \bibnamefont {Hofman}},\ and\
  \bibinfo {author} {\bibfnamefont {S.~F.}\ \bibnamefont {Lokhande}},\ }\href
  {https://doi.org/10.1007/jhep02(2019)058} {\bibfield  {journal} {\bibinfo
  {journal} {J. High Energ. Phys.}\ }\textbf {\bibinfo {volume} {2019}}\bibinfo
   {number} { (2)}}\BibitemShut {NoStop}%
\bibitem [{\citenamefont {Wu}\ and\ \citenamefont
  {Hsieh}(2019)}]{wu2018variational}%
  \BibitemOpen
\bibfield  {number} {  }\bibfield  {author} {\bibinfo {author} {\bibfnamefont
  {J.}~\bibnamefont {Wu}}\ and\ \bibinfo {author} {\bibfnamefont {T.~H.}\
  \bibnamefont {Hsieh}},\ }\href
  {https://doi.org/10.1103/PhysRevLett.123.220502} {\bibfield  {journal}
  {\bibinfo  {journal} {Phys. Rev. Lett.}\ }\textbf {\bibinfo {volume} {123}},\
  \bibinfo {pages} {220502} (\bibinfo {year} {2019})}\BibitemShut {NoStop}%
\bibitem [{\citenamefont {Swingle}\ \emph {et~al.}(2016)\citenamefont
  {Swingle}, \citenamefont {Bentsen}, \citenamefont {Schleier-Smith},\ and\
  \citenamefont {Hayden}}]{Swingle:2016var}%
  \BibitemOpen
  \bibfield  {author} {\bibinfo {author} {\bibfnamefont {B.}~\bibnamefont
  {Swingle}}, \bibinfo {author} {\bibfnamefont {G.}~\bibnamefont {Bentsen}},
  \bibinfo {author} {\bibfnamefont {M.}~\bibnamefont {Schleier-Smith}},\ and\
  \bibinfo {author} {\bibfnamefont {P.}~\bibnamefont {Hayden}},\ }\href
  {https://doi.org/10.1103/PhysRevA.94.040302} {\bibfield  {journal} {\bibinfo
  {journal} {Phys. Rev.}\ }\textbf {\bibinfo {volume} {A94}},\ \bibinfo {pages}
  {040302} (\bibinfo {year} {2016})}\BibitemShut {NoStop}%
\bibitem [{\citenamefont {Yao}\ \emph {et~al.}(2016)\citenamefont {Yao},
  \citenamefont {Grusdt}, \citenamefont {Swingle}, \citenamefont {Lukin},
  \citenamefont {Stamper-Kurn}, \citenamefont {Moore},\ and\ \citenamefont
  {Demler}}]{Yao2016a}%
  \BibitemOpen
  \bibfield  {author} {\bibinfo {author} {\bibfnamefont {N.~Y.}\ \bibnamefont
  {Yao}}, \bibinfo {author} {\bibfnamefont {F.}~\bibnamefont {Grusdt}},
  \bibinfo {author} {\bibfnamefont {B.}~\bibnamefont {Swingle}}, \bibinfo
  {author} {\bibfnamefont {M.~D.}\ \bibnamefont {Lukin}}, \bibinfo {author}
  {\bibfnamefont {D.~M.}\ \bibnamefont {Stamper-Kurn}}, \bibinfo {author}
  {\bibfnamefont {J.~E.}\ \bibnamefont {Moore}},\ and\ \bibinfo {author}
  {\bibfnamefont {E.~A.}\ \bibnamefont {Demler}},\ }\href@noop {} {\bibinfo
  {title} {Interferometric approach to probing fast scrambling}} (\bibinfo
  {year} {2016}),\ \bibinfo {note} {arXiv preprint},\ \Eprint
  {https://arxiv.org/abs/1607.01801} {arXiv:1607.01801} \BibitemShut {NoStop}%
\bibitem [{\citenamefont {Zhu}\ \emph {et~al.}(2016)\citenamefont {Zhu},
  \citenamefont {Hafezi},\ and\ \citenamefont {Grover}}]{Zhu2016}%
  \BibitemOpen
  \bibfield  {author} {\bibinfo {author} {\bibfnamefont {G.}~\bibnamefont
  {Zhu}}, \bibinfo {author} {\bibfnamefont {M.}~\bibnamefont {Hafezi}},\ and\
  \bibinfo {author} {\bibfnamefont {T.}~\bibnamefont {Grover}},\ }\href
  {https://link.aps.org/doi/10.1103/PhysRevA.94.062329} {\bibfield  {journal}
  {\bibinfo  {journal} {Phys. Rev. A}\ }\textbf {\bibinfo {volume} {94}},\
  \bibinfo {pages} {062329} (\bibinfo {year} {2016})}\BibitemShut {NoStop}%
\bibitem [{\citenamefont {Campisi}\ and\ \citenamefont
  {Goold}(2017)}]{Campisi2017}%
  \BibitemOpen
  \bibfield  {author} {\bibinfo {author} {\bibfnamefont {M.}~\bibnamefont
  {Campisi}}\ and\ \bibinfo {author} {\bibfnamefont {J.}~\bibnamefont
  {Goold}},\ }\href {http://arxiv.org/abs/1609.05848
  http://dx.doi.org/10.1103/PhysRevE.95.062127
  http://link.aps.org/doi/10.1103/PhysRevE.95.062127} {\bibfield  {journal}
  {\bibinfo  {journal} {Phys. Rev. E}\ }\textbf {\bibinfo {volume} {95}},\
  \bibinfo {pages} {062127} (\bibinfo {year} {2017})}\BibitemShut {NoStop}%
\bibitem [{\citenamefont {{Yunger Halpern}}(2017)}]{Halpern2016}%
  \BibitemOpen
  \bibfield  {author} {\bibinfo {author} {\bibfnamefont {N.}~\bibnamefont
  {{Yunger Halpern}}},\ }\href {http://arxiv.org/abs/1609.00015
  http://dx.doi.org/10.1103/PhysRevA.95.012120
  https://link.aps.org/doi/10.1103/PhysRevA.95.012120} {\bibfield  {journal}
  {\bibinfo  {journal} {Phys. Rev. A}\ }\textbf {\bibinfo {volume} {95}},\
  \bibinfo {pages} {012120} (\bibinfo {year} {2017})}\BibitemShut {NoStop}%
\bibitem [{\citenamefont {Li}\ \emph {et~al.}(2017)\citenamefont {Li},
  \citenamefont {Fan}, \citenamefont {Wang}, \citenamefont {Ye}, \citenamefont
  {Zeng}, \citenamefont {Zhai}, \citenamefont {Peng},\ and\ \citenamefont
  {Du}}]{Li2017a}%
  \BibitemOpen
  \bibfield  {author} {\bibinfo {author} {\bibfnamefont {J.}~\bibnamefont
  {Li}}, \bibinfo {author} {\bibfnamefont {R.}~\bibnamefont {Fan}}, \bibinfo
  {author} {\bibfnamefont {H.}~\bibnamefont {Wang}}, \bibinfo {author}
  {\bibfnamefont {B.}~\bibnamefont {Ye}}, \bibinfo {author} {\bibfnamefont
  {B.}~\bibnamefont {Zeng}}, \bibinfo {author} {\bibfnamefont {H.}~\bibnamefont
  {Zhai}}, \bibinfo {author} {\bibfnamefont {X.}~\bibnamefont {Peng}},\ and\
  \bibinfo {author} {\bibfnamefont {J.}~\bibnamefont {Du}},\ }\href
  {http://link.aps.org/doi/10.1103/PhysRevX.7.031011} {\bibfield  {journal}
  {\bibinfo  {journal} {Phys. Rev. X}\ }\textbf {\bibinfo {volume} {7}},\
  \bibinfo {pages} {031011} (\bibinfo {year} {2017})}\BibitemShut {NoStop}%
\bibitem [{\citenamefont {Garttner}\ \emph {et~al.}(2017)\citenamefont
  {Garttner}, \citenamefont {Bohnet}, \citenamefont {Safavi-Naini},
  \citenamefont {Wall}, \citenamefont {Bollinger},\ and\ \citenamefont
  {Rey}}]{Garttner2016}%
  \BibitemOpen
  \bibfield  {author} {\bibinfo {author} {\bibfnamefont {M.}~\bibnamefont
  {Garttner}}, \bibinfo {author} {\bibfnamefont {J.~G.}\ \bibnamefont
  {Bohnet}}, \bibinfo {author} {\bibfnamefont {A.}~\bibnamefont
  {Safavi-Naini}}, \bibinfo {author} {\bibfnamefont {M.~L.}\ \bibnamefont
  {Wall}}, \bibinfo {author} {\bibfnamefont {J.~J.}\ \bibnamefont
  {Bollinger}},\ and\ \bibinfo {author} {\bibfnamefont {A.~M.}\ \bibnamefont
  {Rey}},\ }\href {http://arxiv.org/abs/1608.08938
  http://dx.doi.org/10.1038/nphys4119} {\bibfield  {journal} {\bibinfo
  {journal} {Nat. Phys.}\ }\textbf {\bibinfo {volume} {13}},\ \bibinfo {pages}
  {781} (\bibinfo {year} {2017})}\BibitemShut {NoStop}%
\bibitem [{\citenamefont {Meier}\ \emph {et~al.}(2019)\citenamefont {Meier},
  \citenamefont {Ang'ong'a}, \citenamefont {An},\ and\ \citenamefont
  {Gadway}}]{Meier2017}%
  \BibitemOpen
  \bibfield  {author} {\bibinfo {author} {\bibfnamefont {E.~J.}\ \bibnamefont
  {Meier}}, \bibinfo {author} {\bibfnamefont {J.}~\bibnamefont {Ang'ong'a}},
  \bibinfo {author} {\bibfnamefont {F.~A.}\ \bibnamefont {An}},\ and\ \bibinfo
  {author} {\bibfnamefont {B.}~\bibnamefont {Gadway}},\ }\href
  {http://arxiv.org/abs/1705.06714
  https://link.aps.org/doi/10.1103/PhysRevA.100.013623} {\bibfield  {journal}
  {\bibinfo  {journal} {Phys. Rev. A}\ }\textbf {\bibinfo {volume} {100}},\
  \bibinfo {pages} {013623} (\bibinfo {year} {2019})}\BibitemShut {NoStop}%
\bibitem [{\citenamefont {Bernien}\ \emph {et~al.}(2017)\citenamefont
  {Bernien}, \citenamefont {Schwartz}, \citenamefont {Keesling}, \citenamefont
  {Levine}, \citenamefont {Omran}, \citenamefont {Pichler}, \citenamefont
  {Choi}, \citenamefont {Zibrov}, \citenamefont {Endres}, \citenamefont
  {Greiner}, \citenamefont {Vuletic},\ and\ \citenamefont
  {Lukin}}]{Bernien_2017}%
  \BibitemOpen
  \bibfield  {author} {\bibinfo {author} {\bibfnamefont {H.}~\bibnamefont
  {Bernien}}, \bibinfo {author} {\bibfnamefont {S.}~\bibnamefont {Schwartz}},
  \bibinfo {author} {\bibfnamefont {A.}~\bibnamefont {Keesling}}, \bibinfo
  {author} {\bibfnamefont {H.}~\bibnamefont {Levine}}, \bibinfo {author}
  {\bibfnamefont {A.}~\bibnamefont {Omran}}, \bibinfo {author} {\bibfnamefont
  {H.}~\bibnamefont {Pichler}}, \bibinfo {author} {\bibfnamefont
  {S.}~\bibnamefont {Choi}}, \bibinfo {author} {\bibfnamefont {A.~S.}\
  \bibnamefont {Zibrov}}, \bibinfo {author} {\bibfnamefont {M.}~\bibnamefont
  {Endres}}, \bibinfo {author} {\bibfnamefont {M.}~\bibnamefont {Greiner}},
  \bibinfo {author} {\bibfnamefont {V.}~\bibnamefont {Vuletic}},\ and\ \bibinfo
  {author} {\bibfnamefont {M.~D.}\ \bibnamefont {Lukin}},\ }\href
  {https://doi.org/10.1038/nature24622} {\bibfield  {journal} {\bibinfo
  {journal} {Nature}\ }\textbf {\bibinfo {volume} {551}},\ \bibinfo {pages}
  {579} (\bibinfo {year} {2017})}\BibitemShut {NoStop}%
\bibitem [{\citenamefont {Levine}\ \emph {et~al.}(2019)\citenamefont {Levine},
  \citenamefont {Keesling}, \citenamefont {Semeghini}, \citenamefont {Omran},
  \citenamefont {Wang}, \citenamefont {Ebadi}, \citenamefont {Bernien},
  \citenamefont {Greiner}, \citenamefont {Vuleti{\'c}}, \citenamefont {Pichler}
  \emph {et~al.}}]{RydbergBellPairs}%
  \BibitemOpen
  \bibfield  {author} {\bibinfo {author} {\bibfnamefont {H.}~\bibnamefont
  {Levine}}, \bibinfo {author} {\bibfnamefont {A.}~\bibnamefont {Keesling}},
  \bibinfo {author} {\bibfnamefont {G.}~\bibnamefont {Semeghini}}, \bibinfo
  {author} {\bibfnamefont {A.}~\bibnamefont {Omran}}, \bibinfo {author}
  {\bibfnamefont {T.~T.}\ \bibnamefont {Wang}}, \bibinfo {author}
  {\bibfnamefont {S.}~\bibnamefont {Ebadi}}, \bibinfo {author} {\bibfnamefont
  {H.}~\bibnamefont {Bernien}}, \bibinfo {author} {\bibfnamefont
  {M.}~\bibnamefont {Greiner}}, \bibinfo {author} {\bibfnamefont
  {V.}~\bibnamefont {Vuleti{\'c}}}, \bibinfo {author} {\bibfnamefont
  {H.}~\bibnamefont {Pichler}}, \emph {et~al.},\ }\href
  {https://doi.org/10.1103/PhysRevLett.123.170503} {\bibfield  {journal}
  {\bibinfo  {journal} {Phys. Rev. Lett.}\ }\textbf {\bibinfo {volume} {123}},\
  \bibinfo {pages} {170503} (\bibinfo {year} {2019})}\BibitemShut {NoStop}%
\bibitem [{\citenamefont {Bertini}\ \emph {et~al.}(2019)\citenamefont
  {Bertini}, \citenamefont {Kos},\ and\ \citenamefont {Prosen}}]{Bertini_2019}%
  \BibitemOpen
  \bibfield  {author} {\bibinfo {author} {\bibfnamefont {B.}~\bibnamefont
  {Bertini}}, \bibinfo {author} {\bibfnamefont {P.}~\bibnamefont {Kos}},\ and\
  \bibinfo {author} {\bibfnamefont {T.}~\bibnamefont {Prosen}},\ }\href
  {https://doi.org/10.1103/physrevx.9.021033} {\bibfield  {journal} {\bibinfo
  {journal} {Phys. Rev. X}\ }\textbf {\bibinfo {volume} {9}},\ \bibinfo {pages}
  {021033} (\bibinfo {year} {2019})}\BibitemShut {NoStop}%
\bibitem [{\citenamefont {Wright}\ \emph {et~al.}(2019)\citenamefont {Wright},
  \citenamefont {Beck}, \citenamefont {Debnath}, \citenamefont {Amini},
  \citenamefont {Nam}, \citenamefont {Grzesiak}, \citenamefont {Chen},
  \citenamefont {Pisenti}, \citenamefont {Chmielewski}, \citenamefont {Collins}
  \emph {et~al.}}]{wright2019benchmarking}%
  \BibitemOpen
  \bibfield  {author} {\bibinfo {author} {\bibfnamefont {K.}~\bibnamefont
  {Wright}}, \bibinfo {author} {\bibfnamefont {K.}~\bibnamefont {Beck}},
  \bibinfo {author} {\bibfnamefont {S.}~\bibnamefont {Debnath}}, \bibinfo
  {author} {\bibfnamefont {J.}~\bibnamefont {Amini}}, \bibinfo {author}
  {\bibfnamefont {Y.}~\bibnamefont {Nam}}, \bibinfo {author} {\bibfnamefont
  {N.}~\bibnamefont {Grzesiak}}, \bibinfo {author} {\bibfnamefont {J.-S.}\
  \bibnamefont {Chen}}, \bibinfo {author} {\bibfnamefont {N.}~\bibnamefont
  {Pisenti}}, \bibinfo {author} {\bibfnamefont {M.}~\bibnamefont
  {Chmielewski}}, \bibinfo {author} {\bibfnamefont {C.}~\bibnamefont
  {Collins}}, \emph {et~al.},\ }\href
  {https://doi.org/10.1038/s41467-019-13534-2} {\bibfield  {journal} {\bibinfo
  {journal} {Nat. Commun.}\ }\textbf {\bibinfo {volume} {10}},\ \bibinfo
  {pages} {1} (\bibinfo {year} {2019})}\BibitemShut {NoStop}%
\bibitem [{\citenamefont {Davoudi}\ \emph {et~al.}(2020)\citenamefont
  {Davoudi}, \citenamefont {Hafezi}, \citenamefont {Monroe}, \citenamefont
  {Pagano}, \citenamefont {Seif},\ and\ \citenamefont
  {Shaw}}]{davoudi2019analog}%
  \BibitemOpen
  \bibfield  {author} {\bibinfo {author} {\bibfnamefont {Z.}~\bibnamefont
  {Davoudi}}, \bibinfo {author} {\bibfnamefont {M.}~\bibnamefont {Hafezi}},
  \bibinfo {author} {\bibfnamefont {C.}~\bibnamefont {Monroe}}, \bibinfo
  {author} {\bibfnamefont {G.}~\bibnamefont {Pagano}}, \bibinfo {author}
  {\bibfnamefont {A.}~\bibnamefont {Seif}},\ and\ \bibinfo {author}
  {\bibfnamefont {A.}~\bibnamefont {Shaw}},\ }\href
  {https://doi.org/10.1103/PhysRevResearch.2.023015} {\bibfield  {journal}
  {\bibinfo  {journal} {Phys. Rev. Research}\ }\textbf {\bibinfo {volume}
  {2}},\ \bibinfo {pages} {023015} (\bibinfo {year} {2020})}\BibitemShut
  {NoStop}%
\bibitem [{\citenamefont {Zhu}\ \emph {et~al.}(2020)\citenamefont {Zhu},
  \citenamefont {Johri}, \citenamefont {Linke}, \citenamefont {Landsman},
  \citenamefont {Alderete}, \citenamefont {Nguyen}, \citenamefont {Matsuura},
  \citenamefont {Hsieh},\ and\ \citenamefont {Monroe}}]{zhu2019variational}%
  \BibitemOpen
  \bibfield  {author} {\bibinfo {author} {\bibfnamefont {D.}~\bibnamefont
  {Zhu}}, \bibinfo {author} {\bibfnamefont {S.}~\bibnamefont {Johri}}, \bibinfo
  {author} {\bibfnamefont {N.}~\bibnamefont {Linke}}, \bibinfo {author}
  {\bibfnamefont {K.}~\bibnamefont {Landsman}}, \bibinfo {author}
  {\bibfnamefont {C.~H.}\ \bibnamefont {Alderete}}, \bibinfo {author}
  {\bibfnamefont {N.}~\bibnamefont {Nguyen}}, \bibinfo {author} {\bibfnamefont
  {A.}~\bibnamefont {Matsuura}}, \bibinfo {author} {\bibfnamefont
  {T.}~\bibnamefont {Hsieh}},\ and\ \bibinfo {author} {\bibfnamefont
  {C.}~\bibnamefont {Monroe}},\ }\href
  {https://doi.org/10.1073/pnas.2006337117} {\bibfield  {journal} {\bibinfo
  {journal} {Proc. Natl. Acad. Sci. U.S.A.}\ }\textbf {\bibinfo {volume}
  {117}},\ \bibinfo {pages} {25402} (\bibinfo {year} {2020})}\BibitemShut
  {NoStop}%
\bibitem [{\citenamefont {Danshita}\ \emph {et~al.}(2017)\citenamefont
  {Danshita}, \citenamefont {Hanada},\ and\ \citenamefont
  {Tezuka}}]{Danshita_2017}%
  \BibitemOpen
  \bibfield  {author} {\bibinfo {author} {\bibfnamefont {I.}~\bibnamefont
  {Danshita}}, \bibinfo {author} {\bibfnamefont {M.}~\bibnamefont {Hanada}},\
  and\ \bibinfo {author} {\bibfnamefont {M.}~\bibnamefont {Tezuka}},\ }\href
  {https://doi.org/10.1093/ptep/ptx108} {\bibfield  {journal} {\bibinfo
  {journal} {Prog. Theor. Exp. Phys.}\ }\textbf {\bibinfo {volume} {2017}},\
  \bibinfo {pages} {083I01} (\bibinfo {year} {2017})}\BibitemShut {NoStop}%
\bibitem [{\citenamefont {Pikulin}\ and\ \citenamefont
  {Franz}(2017)}]{Pikulin_2017}%
  \BibitemOpen
  \bibfield  {author} {\bibinfo {author} {\bibfnamefont {D.~I.}\ \bibnamefont
  {Pikulin}}\ and\ \bibinfo {author} {\bibfnamefont {M.}~\bibnamefont
  {Franz}},\ }\href {https://doi.org/10.1103/physrevx.7.031006} {\bibfield
  {journal} {\bibinfo  {journal} {Phys. Rev. X}\ }\textbf {\bibinfo {volume}
  {7}},\ \bibinfo {pages} {031006} (\bibinfo {year} {2017})}\BibitemShut
  {NoStop}%
\bibitem [{\citenamefont {Garc\'ia-\'Alvarez}\ \emph
  {et~al.}(2017)\citenamefont {Garc\'ia-\'Alvarez}, \citenamefont {Egusquiza},
  \citenamefont {Lamata}, \citenamefont {del Campo}, \citenamefont {Sonner},\
  and\ \citenamefont {Solano}}]{Garc_a_lvarez_2017}%
  \BibitemOpen
  \bibfield  {author} {\bibinfo {author} {\bibfnamefont {L.}~\bibnamefont
  {Garc\'ia-\'Alvarez}}, \bibinfo {author} {\bibfnamefont {I.~L.}\ \bibnamefont
  {Egusquiza}}, \bibinfo {author} {\bibfnamefont {L.}~\bibnamefont {Lamata}},
  \bibinfo {author} {\bibfnamefont {A.}~\bibnamefont {del Campo}}, \bibinfo
  {author} {\bibfnamefont {J.}~\bibnamefont {Sonner}},\ and\ \bibinfo {author}
  {\bibfnamefont {E.}~\bibnamefont {Solano}},\ }\href
  {https://doi.org/10.1103/physrevlett.119.040501} {\bibfield  {journal}
  {\bibinfo  {journal} {Phys. Rev. Lett.}\ }\textbf {\bibinfo {volume} {119}},\
  \bibinfo {pages} {040501} (\bibinfo {year} {2017})}\BibitemShut {NoStop}%
\bibitem [{\citenamefont {Babbush}\ \emph {et~al.}(2019)\citenamefont
  {Babbush}, \citenamefont {Berry},\ and\ \citenamefont
  {Neven}}]{Babbush_2019}%
  \BibitemOpen
  \bibfield  {author} {\bibinfo {author} {\bibfnamefont {R.}~\bibnamefont
  {Babbush}}, \bibinfo {author} {\bibfnamefont {D.~W.}\ \bibnamefont {Berry}},\
  and\ \bibinfo {author} {\bibfnamefont {H.}~\bibnamefont {Neven}},\ }\href
  {https://doi.org/10.1103/physreva.99.040301} {\bibfield  {journal} {\bibinfo
  {journal} {Phys. Rev. A}\ }\textbf {\bibinfo {volume} {99}},\ \bibinfo
  {pages} {040301} (\bibinfo {year} {2019})}\BibitemShut {NoStop}%
\bibitem [{\citenamefont {Swingle}\ and\ \citenamefont
  {Yunger~Halpern}(2018)}]{Swingle_2018}%
  \BibitemOpen
  \bibfield  {author} {\bibinfo {author} {\bibfnamefont {B.}~\bibnamefont
  {Swingle}}\ and\ \bibinfo {author} {\bibfnamefont {N.}~\bibnamefont
  {Yunger~Halpern}},\ }\href {https://doi.org/10.1103/physreva.97.062113}
  {\bibfield  {journal} {\bibinfo  {journal} {Phys. Rev. A}\ }\textbf {\bibinfo
  {volume} {97}},\ \bibinfo {pages} {062113} (\bibinfo {year}
  {2018})}\BibitemShut {NoStop}%
\bibitem [{\citenamefont {Anderson}\ \emph {et~al.}(2010)\citenamefont
  {Anderson}, \citenamefont {Guionnet},\ and\ \citenamefont
  {Zeitouni}}]{anderson2010introduction}%
  \BibitemOpen
  \bibfield  {author} {\bibinfo {author} {\bibfnamefont {G.~W.}\ \bibnamefont
  {Anderson}}, \bibinfo {author} {\bibfnamefont {A.}~\bibnamefont {Guionnet}},\
  and\ \bibinfo {author} {\bibfnamefont {O.}~\bibnamefont {Zeitouni}},\
  }\href@noop {} {\emph {\bibinfo {title} {An Introduction to Random
  Matrices}}},\ Vol.\ \bibinfo {volume} {118}\ (\bibinfo  {publisher}
  {Cambridge University Press},\ \bibinfo {year} {2010})\BibitemShut {NoStop}%
\bibitem [{\citenamefont {Hayden}\ \emph {et~al.}(2006)\citenamefont {Hayden},
  \citenamefont {Leung},\ and\ \citenamefont {Winter}}]{hayden2006aspects}%
  \BibitemOpen
  \bibfield  {author} {\bibinfo {author} {\bibfnamefont {P.}~\bibnamefont
  {Hayden}}, \bibinfo {author} {\bibfnamefont {D.~W.}\ \bibnamefont {Leung}},\
  and\ \bibinfo {author} {\bibfnamefont {A.}~\bibnamefont {Winter}},\ }\href
  {https://doi.org/10.1007/s00220-006-1535-6} {\bibfield  {journal} {\bibinfo
  {journal} {Commun. Math. Phys.}\ }\textbf {\bibinfo {volume} {265}},\
  \bibinfo {pages} {95} (\bibinfo {year} {2006})}\BibitemShut {NoStop}%
\bibitem [{\citenamefont {Fannes}\ \emph {et~al.}(2004)\citenamefont {Fannes},
  \citenamefont {Haegeman}, \citenamefont {Mosonyi},\ and\ \citenamefont
  {Vanpeteghem}}]{fannes2004additivity}%
  \BibitemOpen
  \bibfield  {author} {\bibinfo {author} {\bibfnamefont {M.}~\bibnamefont
  {Fannes}}, \bibinfo {author} {\bibfnamefont {B.}~\bibnamefont {Haegeman}},
  \bibinfo {author} {\bibfnamefont {M.}~\bibnamefont {Mosonyi}},\ and\ \bibinfo
  {author} {\bibfnamefont {D.}~\bibnamefont {Vanpeteghem}},\ }\href@noop {}
  {\bibinfo {title} {Additivity of minimal entropy output for a class of
  covariant channels}} (\bibinfo {year} {2004}),\ \bibinfo {note} {arXiv
  preprint},\ \Eprint {https://arxiv.org/abs/quant-ph/0410195}
  {arXiv:quant-ph/0410195} \BibitemShut {NoStop}%
\bibitem [{\citenamefont {Datta}\ \emph {et~al.}(2006)\citenamefont {Datta},
  \citenamefont {Holevo},\ and\ \citenamefont {Suhov}}]{datta2006additivity}%
  \BibitemOpen
  \bibfield  {author} {\bibinfo {author} {\bibfnamefont {N.}~\bibnamefont
  {Datta}}, \bibinfo {author} {\bibfnamefont {A.~S.}\ \bibnamefont {Holevo}},\
  and\ \bibinfo {author} {\bibfnamefont {Y.}~\bibnamefont {Suhov}},\ }\href
  {https://doi.org/10.1142/S0219749906001633} {\bibfield  {journal} {\bibinfo
  {journal} {Int. J. Quantum Inf.}\ }\textbf {\bibinfo {volume} {4}},\ \bibinfo
  {pages} {85} (\bibinfo {year} {2006})}\BibitemShut {NoStop}%
\end{thebibliography}%
\clearpage
\appendix
\begin{widetext}

\section{Preliminaries on Pauli operators}\label{sec:pauli}
In this section we review the algebra of $n$-qubit Pauli operators and recall some useful identities.
Consider the Hilbert space $(\CC^2)^{\ot n}$ of an $n$ qubit system.
For any integer vector $\vec v=(\vec p,\vec q)\in\ZZ^{2n}$ we can define a corresponding \emph{Pauli operator} (also known as \emph{Weyl operator}) by
\begin{align}\label{eq:def pauli}
  P_{\vec v}=i^{-\vec p \cdot \vec q } \, Z^{p_1}X^{q_1} \ot \cdots \ot Z^{p_n}X^{q_n}.
\end{align}
The Pauli operators~$P_{\vec v}$ for $v\in\{0,1\}^{2n}$ form a basis of the space of $n$-qubit operators.
However, we caution that $P_{\vec v}$ depends on $\vec v$ modulo $4$ and is well-defined modulo $2$ only up to a sign.
Namely,
\begin{align}\label{eq:mod2mod4}
  P_{\vec v +2\vec w} = (-1)^{[\vec v,\vec w]}P_{\vec v},
\end{align}
where $[\cdot,\cdot]$ is the `symplectic form' defined by $[(\vec p,\vec q),(\vec p',\vec q')] = \vec p\cdot \vec q' - \vec q\cdot \vec p'$.
Using this form, the commutation relation of the Pauli operators can be succinctly written as
\begin{align*}
  P_{\vec v} P_{\vec w} = (-1)^{[\vec v,\vec w]} P_{\vec w} P_{\vec v}
\end{align*}
and multiplication is given by
\begin{align}\label{eq:group_str}
  P_{\vec v} P_{\vec w} =i^{[\vec v,\vec w]} P_{\vec v+ \vec w},
\end{align}
where the addition $\vec v+\vec w$ in $P_{\vec v+ \vec w}$ must be carried out modulo~$4$ and can only be reduced to the range~$\{0,1\}$ by carefully applying~\cref{eq:mod2mod4}.
Finally, we note that the transpose of a Pauli operator is given by
\begin{align}\label{eq:Pauli_Transpose}
P_{\vec v}^T = (-1)^{\vec p\cdot\vec q} P_{\vec v},
\end{align}
since transposing only impacts the $Y$ operators.

With these facts in mind, let us discuss the size of Pauli operators.
The \emph{size} (or \emph{weight}) of a Pauli operator~$P = P_{\vec v}$, which we denote by $\abs{P} = \abs{\vec v}$, is defined as the number of single-qubit Paulis in \cref{eq:def pauli} that are not proportional to an identity operator.
The locations of those Pauli operators are called the \emph{support} of~$P$, which is a subset of $\{1,\dots,n\}$.
If $\vec v \in \{0,1\}^{2n}$ then the size of $P_{\vec v}$ can be calculated as $\vec p \cdot \vec p + \vec q \cdot \vec q - \vec p \cdot \vec q$, where the last term ensures we do not double count the~$Y$ operators.
Using the properties above, we arrive at an identity which holds for all $\vec v\in\ZZ^{2n}$ and will frequently be used:
\begin{align}\label{eq:Ytranspose_Pauli}
  Y^{\ot n} P_{\vec v}^T Y^{\ot n}
= (-1)^{\vec p \cdot \vec p + \vec q \cdot \vec q} P_{\vec v}^T
= (-1)^{\abs{P_\vec v}} P_{\vec v}.
\end{align}

\section{Proof of Eq.~(\ref{eq:TS_Explicit})}\label{app:state transfer by size dependent phase}
In this section, we derive the formula for the teleportation-by-size channel.
Clearly it suffices to prove \cref{eq:TS_Explicit} for $m=1$, since both the left-hand and the right-hand side are tensor power channels.
Thus we need to show that if $S=S_{A_LA_R}$ is the two-qubit unitary that acts as
\begin{align*}
  S\ket{P} = e^{ig\abs{P}} \ket{P}
\end{align*}
for all single-qubit Paulis $P$ then we have that
\begin{align}\label{eq:TS_Explicit single qubit}
  \Tr_{A_L}[S(\rho \ot \tau)S^\dagger]
= Y \Delta_\lambda(\rho) Y
\end{align}
for any single-qubit state~$\rho$, where $\tau=I/2$ is the maximally mixed state, $\Delta_\lambda$ is the single-qubit depolarizing channel $\Delta_\lambda(\rho) = (1-\lambda)\tau + \lambda\rho$, and $\lambda=(1-\cos(g))/2$.

To verify \cref{eq:TS_Explicit single qubit}, note that we can write
\begin{align*}
  S
= \phi^+ + e^{ig} \left( I - \phi^+ \right)
= e^{ig} \left( I + (e^{-ig} - 1) \phi^+ \right),
\end{align*}
where $\phi^+ = \proj{\phi^+}$ denotes the projector onto the maximally entangled state $\ket{\phi^+}=\ket I = (\ket{00} + \ket{11})/\sqrt2$.
Thus,
\begin{align*}
S(\rho \ot \tau)S^\dagger
&= \rho \ot \tau + (e^{-ig} - 1) \phi^+ (\rho \ot \tau) + (e^{ig} - 1) (\rho \ot \tau) \phi^+ + \frac{(e^{-ig} - 1)(e^{ig} - 1)}{4} \phi^+ \\
&= \rho \ot \tau + \frac{e^{-ig} - 1}2 \phi^+ (I \ot \rho^T) + \frac{e^{ig} - 1}2 (I \ot \rho^T) \phi^+ + \frac{(e^{-ig} - 1)(e^{ig} - 1)}{4} \phi^+
\end{align*}
using the transpose trick, and hence
\begin{align*}
  \Tr_{A_L}[S(\rho \ot \tau)S^\dagger]
&= \tau + \frac{e^{-ig} - 1}4 \rho^T + \frac{e^{ig} - 1}4 \rho^T + \frac{(e^{-ig} - 1)(e^{ig} - 1)}{4} \tau \\
&= (1 + \lambda) \tau - \lambda \rho^T
= (1 - \lambda) \tau + \lambda Y \rho Y
= Y \Delta_\lambda(\rho) Y,
\end{align*}
since, for qubits, $\rho^T = I - Y \rho Y$.
This proves \cref{eq:TS_Explicit single qubit} and, hence, \cref{eq:TS_Explicit}.

\section{Random unitary time evolution}\label{app:random unitaries}
In this section we establish our technical results for a random unitary time evolution, which were discussed in \cref{subsec:Teleportation_by_Size}.

\subsection{Proof of Eq.~(\ref{eq:sandwich haar}) and concentration}\label{sec:avg_sec}
We first compute the average ``sandwiched'' coupling in \cref{eq:sndc} in case the time evolution is given by a random unitary, that is,
\begin{align*}
  M_{LR} = \EE_U \tilde M_{LR}, \quad \text{ where } \quad \tilde M_{LR} = ( \overline U_L \ot U_R ) e^{igV} ( U_L^T\ot U_R^\dagger ).
\end{align*}
Consider the partial transpose $M_{LR}^{T_R}$ of $M_{LR}$ on the right subsystem.
Using the invariance property of the Haar measure, one can see that for every unitary $V$,
\begin{align*}
  (V_L \ot V_R) M_{LR}^{T_R} = M_{LR}^{T_R} (V_L \ot V_R).
\end{align*}
By Schur-Weyl duality, any operator that commutes with all matrices of the form $V^{\ot r}$ is in the span of the permutations of~$r$ replicas of the system. Here, we have $r=2$, and there exist only two permutations: the identity and the flip operator.
Hence, we have:
\begin{align*}
  (M_{LR})^{T_R} = \alpha' I_{LR} +\beta' F_{LR}, \text{ for some }\alpha'\text{  and }\beta'.
\end{align*}
Now, note that $(I_{LR})^{T_R}=I_{LR}$ and $(F_{LR})^{T_R} \propto \phi^+_{LR}$. By taking another partial transpose of the above equation we get:
\begin{align*}
  M_{LR} =  \alpha I_{LR} + \beta \phi^+_{LR}.
\end{align*}
To determine the coefficients $\alpha$ and $\beta$, we compute the following traces:
\begin{align*}
  4^n \alpha + \beta &= \tr (M_{LR}) = \tr(e^{igV}) = 4^n \cos(g/k)^k, \\
  \alpha  + \beta &= \tr (M_{LR} \, \phi^+_{LR}) = \tr(e^{igV}\phi^+_{LR}) = e^{ig}.
\end{align*}
Therefore:
\begin{align*}
  \alpha &= \frac {\cos(g/k)^k - 4^{-n} e^{ig}}{1 - 4^{-n}} = \cos(g/k)^k + O(4^{-n}) \\
  \beta &= \frac {e^{ig} - \cos(g/k)^k}{1 - 4^{-n}} = e^{ig} - \cos(g/k)^k + O(4^{-n})
\end{align*}
Therefore, we obtain:
\begin{align}
 M_{LR} = \cos(g/k)^k (I_{LR} - \phi^+_{LR}) + e^{ig} \phi^+_{LR} + O(4^{-n})
\end{align}
up to corrections of order $O(4^{-n}$).
This establishes \cref{eq:sandwich haar} in the main text.

We now consider the projection of the ``sandwiched'' coupling onto a maximally entangled state on the carrier qubits, that is,
\begin{align}\label{eq:projected coupling}
  \tilde S_{A_L A_R} = \braum_{B_L B_R} (\overline U_L \ot U_R) e^{igV} (U_L^T\ot U_R^\dagger) \ketum_{B_L B_R}.
\end{align}
Clearly,
\begin{align}
  \EE( \tilde S_{A_L A_R})
&= \braum_{B_L B_R} M_{LR} \ketum_{B_L B_R} \nonumber \\
&= \cos(g/k)^k (I_{A_LA_R} - \phi^+_{A_LA_R}) + e^{ig} \phi^+_{A_LA_R} + O(4^{-n}). \label{eq:average_formula_in_app}
\end{align}
We now prove that $\tilde S_{A_LA_R}$ concentrates around its average.
There are many different ways to prove this, such as by computing the variance directly or using the fact that if the average of a number of unitaries is close to a unitary, then the distribution must be peaked near its average.
Here, we chose to use a slightly more technical approach employing Levy's lemma.
This has the benefit of giving a generalizable proof technique with stronger bounds.
Levy's lemma for the unitary group~\cite[Corollary~4.4.28]{anderson2010introduction} states that if a function $f\colon U(2^n) \to \RR$ is $\lambda$-Lipschitz, meaning that
\begin{align*}
  \abs{f(U)-f(V)} \leq \lambda \norm{U - V}_F \quad \text{for all } U, V \in U(2^n),
\end{align*}
then
\begin{align}\label{eq:levy}
  \Pr\Bigl( \abs{f(U) - \EE f(U)} \geq \epsilon \Bigr) \leq 2\exp\Bigl( - 2^n \frac{\eps^2}{4\lambda^2} \Bigr).
\end{align}
We will first bound the matrix elements of $\tilde M_{LR}$.
For this, consider the function
\begin{align*}
  f_\Psi \colon U(2^n) \to \RR, \quad f_\Psi(U)
= \braket{\Psi | \tilde M_{LR} | \Psi}
= \braket{\Psi | (\overline U_L \ot U_R) e^{igV} (U_L^T\ot U_R^\dagger) | \Psi}
\end{align*}
for any fixed pure state~$\ket\Psi$.
We have
\begin{align*}
  \abs{f_\Psi(U) - f_\Psi(V)}
&\leq \norm{(\overline U_L \ot U_R) e^{igV} (U_L^T \ot U_R^\dagger) - (\overline V_L \ot V_R) e^{igV} (V_L^T\ot V_R^\dagger)}_{\operatorname{op}} \\
&\leq \norm{(\overline U_L \ot U_R) e^{igV} (U_L^T \ot U_R^\dagger) - (\overline U_L \ot U_R) e^{igV} (U_L^T\ot V_R^\dagger)}_{\operatorname{op}} \\
&+ \norm{(\overline U_L \ot U_R) e^{igV} (U_L^T \ot V_R^\dagger) - (\overline U_L \ot U_R) e^{igV} (V_L^T\ot V_R^\dagger)}_{\operatorname{op}} \\
&+ \norm{(\overline U_L \ot U_R) e^{igV} (V_L^T \ot V_R^\dagger) - (\overline U_L \ot V_R) e^{igV} (V_L^T\ot V_R^\dagger)}_{\operatorname{op}} \\
&+ \norm{(\overline U_L \ot V_R) e^{igV} (V_L^T \ot V_R^\dagger) - (\overline V_L \ot V_R) e^{igV} (V_L^T\ot V_R^\dagger)}_{\operatorname{op}} \\
&= \norm{e^{igV} (U_R^\dagger - V_R^\dagger)}_{\operatorname{op}}
+ \norm{e^{igV} (U_L^T - V_L^T)}_{\operatorname{op}} \\
&+ \norm{(U_R - V_R) e^{igV}}_{\operatorname{op}}
+ \norm{(\overline U_L - \overline V_L) e^{igV}}_{\operatorname{op}} \\
&\leq 4 \norm{U - V}_{\operatorname{op}} \leq 4 \norm{U - V}_F.
\end{align*}
Thus, $f_\Psi$ is 4-Lipschitz and \cref{eq:levy} shows that
\begin{align}\label{eq:levy matrix element}
  \Pr\Bigl( \abs{f_\Psi(U) - \EE f_\Psi(U)} \geq \epsilon \Bigr) \leq 2\exp\Bigl( - 2^n \frac{\epsilon^2}{64} \Bigr).
\end{align}
This result implies very strong concentration of all the individual matrix elements of~$M$, but we would like to prove concentration of the whole matrix $\tilde S_{A_LA_R}$ in the operator norm.
For this we use the existence of small $\epsilon$-nets.
Namely, it is known that there exists a set $\mathcal N$ of at most $(5/\epsilon)^{2\times 4^m}$ many pure states in the $4^m$-dimensional $A_L\ot A_R$ satisfying the following property~\cite{hayden2006aspects}:
\begin{quote}
  \emph{For every pure state~$\Psi_{A_LA_R}$, there exists $\tilde\Psi_{A_LA_R} \in \mathcal N$ such that $\norm{ \Psi_{A_LA_R} - \tilde\Psi_{A_LA_R} }_{\operatorname{tr}} \leq \epsilon$.}
\end{quote}
Now,
\begin{align*}
  \norm{ \tilde S_{A_LA_R} - \EE(\tilde S_{A_LA_R}) }_{\operatorname{op}}
= \max_{\Psi_{A_LA_R} \text{ pure}} \Tr \Psi_{A_LA_R} (\tilde S_{A_LA_R} - \EE(\tilde S_{A_LA_R}) )
\end{align*}
If $\tilde\Psi_{A_LA_R}$ is the element of the net $\mathcal N$ closest to some $\Psi_{A_LA_R}$ in trace norm then
\begin{align*}
&\quad \Tr \Psi_{A_LA_R} (\tilde S_{A_LA_R} - \EE(\tilde S_{A_LA_R}) ) \\
&\leq \Tr \tilde\Psi_{A_LA_R} (\tilde S_{A_LA_R} - \EE(\tilde S_{A_LA_R}) )
+ \norm{\Psi_{A_LA_R} - \tilde\Psi_{A_LA_R}}_{\operatorname{tr}} \norm{\tilde S_{A_LA_R} - \EE(\tilde S_{A_LA_R})}_{\operatorname{op}} \\
&\leq \Tr \tilde\Psi_{A_LA_R} (\tilde S_{A_LA_R} - \EE(\tilde S_{A_LA_R}) ) + 2\eps \\
&= f_{\tilde\Psi_{A_LA_R} \ot \phi^+_{B_LB_R}}(U) - \EE f_{\tilde\Psi_{A_LA_R} \ot \phi^+_{B_LB_R}}(U) + 2\eps
\end{align*}
Therefore,
\begin{align*}
  \norm{ \tilde S_{A_LA_R} - \EE(\tilde S_{A_LA_R}) }_{\operatorname{op}}
\leq \max_{\tilde\Psi_{A_LA_R} \in \mathcal N} f_{\tilde\Psi_{A_LA_R} \ot \phi^+_{B_LB_R}}(U) - \EE f_{\tilde\Psi_{A_LA_R} \ot \phi^+_{B_LB_R}}(U) + 2\eps
\end{align*}
Therefore:
\begin{align*}
&\quad \Pr\Bigl( \norm{ \tilde S_{A_LA_R} - \EE(\tilde S_{A_LA_R}) }_{\operatorname{op}} \geq 3\eps \Bigr) \\
&\leq \Pr\Bigl(\exists \tilde\Psi_{A_LA_R} \in \mathcal N: \abs{f_{\tilde\Psi_{A_LA_R} \ot \phi^+_{B_LB_R}}(U) - \EE f_{\tilde\Psi_{A_LA_R} \ot \phi^+_{B_LB_R}}(U)} \geq \eps) \\
&\leq 2 \left(\frac5\eps\right)^{2\times 4^m} \exp\Bigl( - 2^n \frac{\epsilon^2}{64} \Bigr)
= 2 \exp\Bigl( 2^{2m+1} \log\frac5\eps - 2^n \frac{\epsilon^2}{64} \Bigr).
\end{align*}
where the second inequality follows from \cref{eq:levy matrix element} and the union bound.
This shows that for $n\gg m$ the operator $\tilde S_{A_LA_R}$ is with very high probability very close to its mean.

\subsection{Sending many qubits: The transpose depolarizing channel}\label{app:Transpose_depo}
In \cref{sec:avg_sec} we computed the effective coupling between the left and right-hand side message subsystems.
Using \cref{eq:average_formula_in_app}, one can see that for~$n\gg m$ the net effect of the teleportation protocol is with high probability given by the following formula, up to $O(4^{-m})$ corrections:
\begin{align}\label{eq:teleportation_channel}
  \rho \mapsto \tau + \frac4{d^2} \sin(g/2)^2 \Bigl(\tau - \rho^T\Bigr),
\end{align}
where $d = 2^m$ and $\tau = I/d$ denotes the maximally mixed state.
We can interpret $\tau$ as noise, $(\tau - \rho^T)$ as the signal, and the factor~$(4/d^2) \sin(g/2)^2$ as the strength of the signal.
The latter attains its maximum $4/d^2$ only for $g = \pi \pmod{2\pi}$.

For a single qubit, and only in this case, high fidelity teleportation is possible.
Indeed, for $d=2$ the channel
\begin{align*}
  \rho \mapsto \tau + (\tau - \rho^T) = I - \rho^T = Y \rho Y,
\end{align*}
is unitary (cf.\ \cref{eq:TS_Explicit}), while for $d>2$ the signal is suppressed by $4/d^2\leq4/9$.

There another reason that shows that reliably transmitting more than one qubit by the single use of a channel of the above form is impossible.
The mapping $\rho \mapsto \tau + (\tau-\rho^T)$ is not a valid quantum channel for $d>2$.
In fact, the mapping
\begin{align*}
  \rho \mapsto \tau + \alpha (\tau - \rho^T)
\end{align*}
is completely positive (and hence a valid quantum channel) if and only if $-1/(d+1) \leq \alpha \leq 1/(d-1)$.
The channel in \cref{eq:teleportation_channel} is an example of the class of \emph{transpose depolarizing channels}~\cite{fannes2004additivity}, which were originally studied as examples of channels where the minimal output entropy is additive~\cite{fannes2004additivity,datta2006additivity}.

\subsection{General teleportation channel at infinite temperature}\label{app:general_channel}
So far, we considered teleportation protocols where the $m$~message qubits $A_L$ are replaced by the state that is to be teleported.
Here we will show that even when using a very general communication protocol there are fundamental limitations, and even when we only wish to communicate classical bits.
Specifically, we will consider a protocol where one applies one out of several arbitrary quantum channels to~$A_L$ (one for each possible message).

Let us first study the effect of applying a single quantum channel $\mathcal N$ on the message system $A_L$.
In this case, the output state on system~$A_R$ is given by the following circuit diagram:
\begin{center}
\includegraphics[width=14 cm]{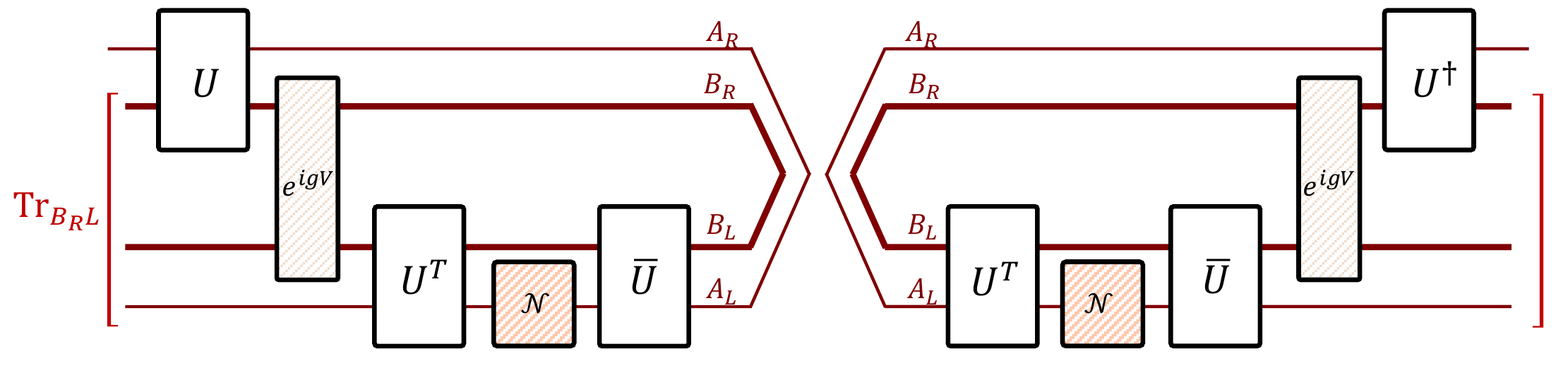}
\end{center}
It is known that one can always write a quantum channel as the following process:
(1) add an environment system in a fixed quantum state,
(2) evolve the environment and the system by some unitary, and
(3) trace out the environment.
This is known as the Stinespring representation.
If we represent our quantum channel in this form, $\mathcal N(\rho_{A_L}) = \tr_{E}\left(G_{A_L E}\, (\rho_{A_L} \ot \ket0\bra0) \,G^\dagger_{A_L E} \right)$, we can write:
\begin{center}
\includegraphics[width=14 cm]{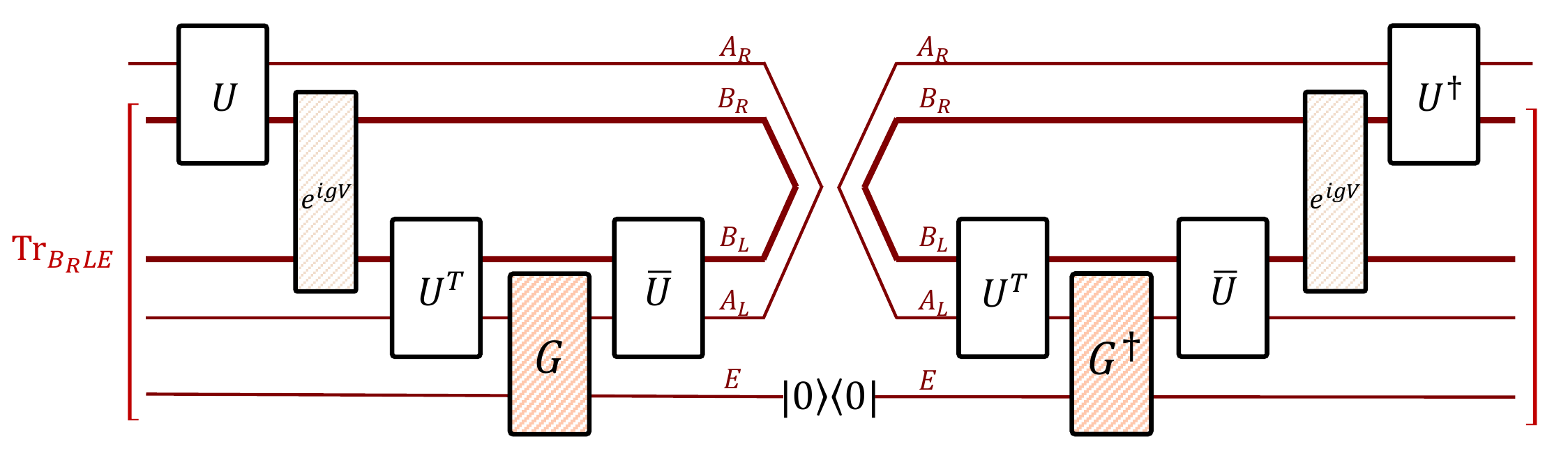}
\end{center}
Finally, after using transpose trick we obtain the following equivalent diagram:
\begin{center}
\includegraphics[width=14 cm]{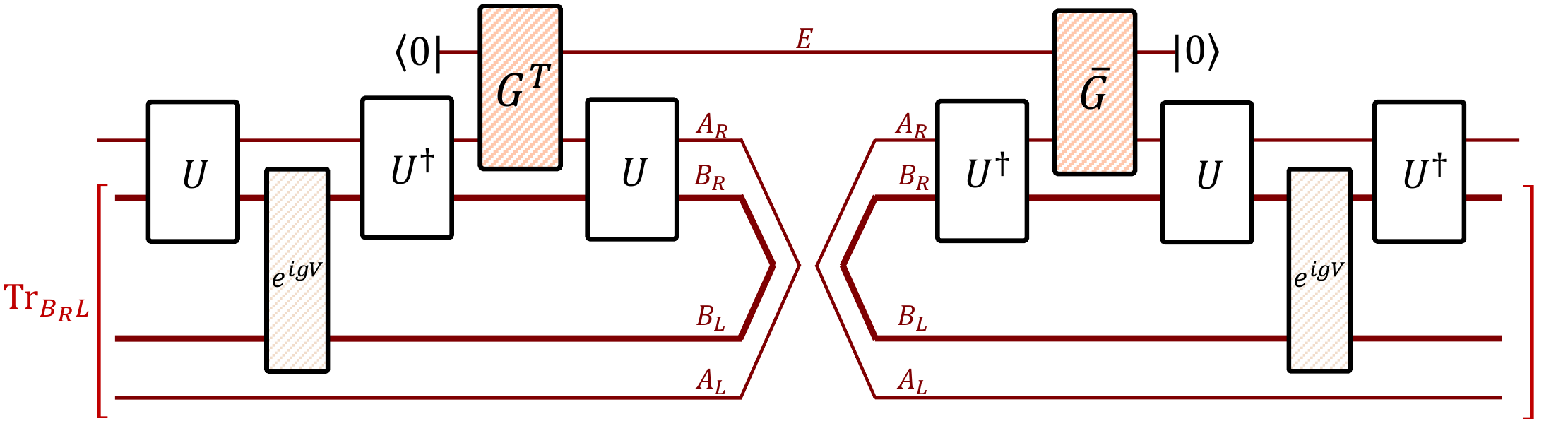}
\end{center}
It is possible to compute the average output state exactly by using the above diagram using a Haar integration similar to the above.
The final result is that for any channel $\mathcal N$, the average output state on $A_R$ is given by the following formula up to $O(4^{-n})$ corrections:
\begin{align}\label{eq:generic_channel_formula}
   \rho_{A_R}
&= \tau - \Big[ \big(1-e^{ig}\big) \big(Q-\tr[Q]\tau\big) + \mathrm{h.c.} \Big] \\
\mathrm{where }\quad Q &= \tr_{A_L} \left[\mathcal N (\phi^+_{A_LA_R})\phi^+_{A_LA_R} \right]\nonumber
\end{align}
Moreover, the output state is self-averaging, i.e., close to its average with high probability.
One can see that when $\mathcal N$ is the channel that replaces a qubit by a new one in a fixed state, then the output signal depends on $g$ through $\cos(g)$, while when $\mathcal N (\rho_{A_L})=e^{-i \eps O} \rho_{A_L} e^{-i \eps O}$ for some Hermitian operator $O$, then the signal depends on $\sin(g)$.

Next, we will show that the output state in \cref{eq:generic_channel_formula} is highly mixed.
For this it is useful to consider the Kraus representation of the channel, i.e., $\mathcal N(\rho_{A_L}) = \sum_i E_i \rho_{A_L} E_i^\dagger$.
Then, one can check that
\begin{align*}
  Q = \sum_i E_i^T \frac{\tr E_i^\dagger}{d^2}
\end{align*}
Consider the following matrix inequality, which holds for arbitrary $x\in\CC$:
\begin{align*}
  \frac1d \sum_i (E_i + \bar x \tr[E_i] \tau)^\dagger (E_i + \bar x \tr[E_i] \tau) \geq 0.
\end{align*}
Since $\mathcal N$ is trace-preserving, $\sum_i E_i^\dagger E_i = I$, and now a short calculation shows that the preceding matrix inequality is equivalent to $\tau + \abs x^2 \tr[Q] \tau \geq -x Q -\bar x Q^\dagger$,
and hence
\begin{align*}
  \tau \left( 1 + (\abs x^2 + x + \bar x) \Tr[Q] \right) \geq -x (Q - \Tr[Q] \tau) - \bar x (Q - \Tr[Q] \tau)^\dagger.
\end{align*}
Setting $x = (1-e^{ig})/2$ and using \cref{eq:generic_channel_formula}, we obtain
\begin{align}\label{eq:mixed bound}
  \rho_{A_R}
\leq 3 \Bigl( 1 + (1-\cos(g)) \braket{\phi^+ | \mathcal N(\phi^+) | \phi^+} \Bigr) \tau
\leq 9 \tau.
\end{align}
This is a strong constraint, since it implies that all the eigenvalues of~$\rho_{A_R}$ are smaller than $9/d$.
This in turn means that at least~$d/9$ eigenvalues are nonzero.
Therefore, one can send at most $9$ perfectly distinguishable states in this way.

One can also compute the Holevo information of an ensemble~$\{p_i,\rho_i\}$, where each state~$\rho_i$ is the output state for a different choice of channel~$\mathcal N_i$.
Recall that the Holevo information is defined as
\begin{align*}
  \chi(\{\rho_i,p_i\}) = S(\sum_i p_i \rho_i) - \sum_i p_i S(\rho_i).
\end{align*}
But $S(\sum_i p_i \rho_i) \leq \log d$, and if $\rho_i \leq c \tau$ then $S(\rho_i) \geq \log d - \log c$.
Accordingly, for any ensemble of states satisfying \cref{eq:mixed bound} we can bound the Holevo information as
\begin{align*}
  \chi(\{\rho_i,p_i\}) \leq \log 9 \approx 3.17,
\end{align*}
where we use the logarithm to base 2.

\section{Proof of Eq.~(\ref{eq:eraf})}\label{app:Twisting}
In this section, we show that the action of the weak coupling unitary $e^{igV}$ on the state $O_R \ketum_{LR}$ amounts to approximately a size-dependent phase under a natural assumption on the operator~$O$.

Recall the setup in the main text: $L$ and~$R$ are each split into an $m$-qubit `message' subsystem~$A$ and a $k$-qubit `carrier' subsystem~$B$, where $m+k=n$.
The two sides are coupled by the Hamiltonian $V = \frac1k \sum_{i \in B} Z_i^L Z_i^R$.
It is straightforward to see that
\begin{align*}
  e^{igV} \ket P_{LR}
= e^{ig\bigl( 1 - 2\frac{\abs P_X^B}k \bigr)} \ket P_{LR},
\end{align*}
where $\abs P_X^B$ denotes the number of single-qubit Pauli~$X$ or~$Y$ operators in~$P$ that act on the $B$~subsystem.
Therefore, we have for any operator $O = 2^{-n/2} \sum_P c_P P$ expanded in the Pauli basis that
\begin{align}\label{eq:exact action}
  e^{igV} O_R \ketum_{LR}
= O_R^{(g)_X^B} \ketum_{LR},
\end{align}
where we defined
\begin{align}\label{eq:X B twist}
  O^{(g)_X^B} := 2^{-n/2} \sum_P e^{ig\bigl( 1 - 2\frac{\abs P_X^B}k \bigr)} c_P P.
\end{align}
For typical Pauli operators~$P$ and~$n\gg m$, it holds that $2\frac{\abs P_X^B}k \approx \frac43 \frac{\abs P}n$, suggesting one might be able to replace \cref{eq:X B twist} by the following operator:
\begin{align}\label{eq:twist}
  O^{(g)} := 2^{-n/2} \sum_P e^{ig\bigl( 1 - \frac43 \frac{\abs P}n \bigr)} c_P P.
\end{align}
The following lemma shows that this is indeed valid under the natural assumption that the coefficients $\abs{c_P}^2$ only depend on the support of~$P$.

\begin{lem}\label{lem:difftwist}
Let $O = 2^{-n} \sum_P c_P P$ be an operator such that $\abs{c_P}^2 = \abs{c_{P'}}^2$ for any two Pauli operators~$P,P'$ with equal support.
Then:
\begin{align*}
  \norm[\big]{O^{(g)_X^B} - O^{(g)}}_{F}^2
\leq \frac{4}3 g \sqrt{ \frac 1{2k} + \left( \frac m n \right)^2} \norm O_F.
\end{align*}
\end{lem}

\noindent
The right-hand side is negligible if $\norm O_F = O(1)$, $g^2 \ll k$, $gm \ll n$.
In this case, \cref{eq:exact action} also implies that
\begin{align*}
  e^{igV} O_R \ketum_{LR} \approx O_R^{(g)} \ketum_{LR},
\end{align*}
since the Frobenius norm dominates the operator norm.
This establishes \cref{eq:eraf} in the main text.

\begin{proof}
We start by expanding
\begin{align*}
&\norm[\big]{O^{(g)_X^B} - O^{(g)}}_{F}^2 \\
&= 2^{-n} \norm[\big]{\sum_P c_P e^{ig\bigl( 1 - 2\frac{\abs P_X^B}k \bigr)} P - \sum_P c_P e^{ig\bigl( 1 - \frac43 \frac{\abs P}n \bigr)} P}_{F}^2 \\
&= \sum_P \abs{c_P}^2 \abs[\big]{e^{ig\bigl( 1 - 2\frac{\abs P_X^B}k \bigr)} - e^{ig\bigl( 1 - \frac43 \frac{\abs P}n \bigr)}}^2 \\
&\leq \sum_P \abs{c_P}^2 \abs[\big]{e^{ig\bigl( 1 - 2\frac{\abs P_X^B}k \bigr)} - e^{ig\bigl( 1 - \frac43 \frac{\abs P^B}k \bigr)}}^2
+ \sum_P \abs{c_P}^2 \abs[\big]{e^{ig\bigl( 1 - \frac43 \frac{\abs P^B}k \bigr)} - e^{ig\bigl( 1 - \frac43 \frac{\abs P}n \bigr)}}^2
\end{align*}
The first term can be bounded using our assumption, as follows.
Let $\EE_{P' \sim P}$ denote the uniform average over all Pauli operator~$P'$ that have the same support as some given Pauli operator~$P$ (that is, $P'_i=I$ if $P_i=I$, otherwise $P'_i$ is chosen independently and uniformly from $\{X,Y,Z\}$).
Then:
\begin{align*}
&\quad \sum_P \abs{c_P}^2 \abs[\big]{e^{ig\bigl( 1 - 2\frac{\abs P_X^B}k \bigr)} - e^{ig\bigl( 1 - \frac43 \frac{\abs P^B}k \bigr)}}^2
= \sum_P \abs{c_P}^2 \abs[\big]{e^{i\frac{2g}k\bigl(\frac23 \abs P^B - \abs P_X^B \bigr)} - 1}^2 \\
&\leq \left( \frac{2g}k \right)^2 \sum_P \abs{c_P}^2 \abs*{\frac23 \abs P^B - \abs P_X^B}^2
= \left( \frac{2g}k \right)^2 \sum_P \abs{c_P}^2 \EE_{P' \sim P} \abs*{\frac23 \abs P^B - \abs {P'}_X^B}^2 \\
&= \left( \frac{2g}k \right)^2 \sum_P \abs{c_P}^2 \operatorname{Var}_{P' \sim P}(\abs{P'}_X^B)
= \left( \frac{2g}k \right)^2 \sum_P \abs{c_P}^2 \frac29 \abs{P}^B
\leq \frac{8g^2}{9k} \sum_P \abs{c_P}^2,
\end{align*}
where we first used using $\abs{e^{it} - 1}^2 \leq t^2$ and then the assumption.
The second term can be bounded as follows without using the assumption:
\begin{align*}
&\quad \sum_P \abs{c_P}^2 \abs[\big]{e^{ig\bigl( 1 - \frac43 \frac{\abs P^B}k \bigr)} - e^{ig\bigl( 1 - \frac43 \frac{\abs P}n \bigr)}}^2
= \sum_P \abs{c_P}^2 \abs[\big]{e^{ig\frac43 \bigl(\frac{\abs P}n - \frac{\abs P^B}k\bigr)} - 1}^2 \\
&\leq \left( g \frac43 \right)^2 \sum_P \abs{c_P}^2 \left( \frac{\abs P}n - \frac{\abs P^B}k \right)^2
= \left( g \frac43 \right)^2 \sum_P \abs{c_P}^2 \left( \frac{\abs P^A}n - \frac mn \frac{\abs P^B}k \right)^2
\leq \left( g \frac43 \frac m n \right)^2 \sum_P \abs{c_P}^2
\end{align*}
Since $\norm O_F^2 = \sum_P \abs{c_P}^2$, we obtain the desired result.
\end{proof}

\section{Proof of Eq.~(\ref{eq:qtildemaintext})}\label{app:two point}
For an arbitrary operator $O$, consider
\begin{align}\label{eq:fourier defined as trace}
  \tilde q_O(g) = e^{-ig} \tr\left[ \rho_\beta^{1/2} O(t) \left(\rho_\beta^{1/2} O(t) \right)^{(g)} \right],
\end{align}
using the notation defined in \cref{eq:twist} and $O(t) = U^\dagger O U$.

This quantity can be interpreted as the Fourier transform of the winding size distribution of~$\rho_\beta^{1/2} O(t)$.
Recall the latter is defined as $q_O(l) = \sum_{\abs P = l} c_P^2$, where $\rho_\beta^{1/2} O(t) = 2^{-n/2} \sum_P c_P P$.
Therefore,
\begin{align*}
  \tilde q_O(g)
= e^{-ig} \tr\left[ \rho_\beta^{1/2} O(t) \left(\rho_\beta^{1/2} O(t) \right)^{(g)} \right]
= \sum_P c_P^2 e^{-ig\frac43 \frac{\abs P}n}
= \sum_{l=0}^n q_O(l) e^{-ig\frac43 \frac ln}
\end{align*}

We wish to compare \cref{eq:fourier defined as trace} with the following two-point function
\begin{align}\label{eq:two point}
\bra{\TFD} O_R(t) e^{igV} O_L^T(-t) \ket{\TFD}
&= \Tr\left[ \rho_\beta^{1/2} O(t) \left( \rho_\beta^{1/2} O(t) \right)^{(g)_X^B} \right],
\end{align}
where the equality follows using $\ket{\TFD} := 2^{n/2}(\rho_\beta^{1/2})_R\ketum$, the transpose trick, and \cref{eq:exact action}.
Assuming the thermal operator $\rho_\beta^{1/2} O(t)$ satisfies the hypotheses of \cref{lem:difftwist}, we can the two-point function to \cref{eq:fourier defined as trace}:
\begin{align*}
&\quad\abs*{\Tr\left[ \rho_\beta^{1/2} O(t) \left( \rho_\beta^{1/2} O(t) \right)^{(g)_X^B} \right] - \tr\left[ \rho_\beta^{1/2} O(t) \left(\rho_\beta^{1/2} O(t) \right)^{(g)} \right]} \\
&\leq \norm*{\rho_\beta^{1/2} O(t)}_F \norm*{\left( \rho_\beta^{1/2} O(t) \right)^{(g)_X^B} - \left(\rho_\beta^{1/2} O(t) \right)^{(g)}}_F \\
&\leq \frac{4}3 g \sqrt{ \frac 1{2k} + \left( \frac m n \right)^2} \norm*{\rho_\beta^{1/2} O(t)}_F^2
\leq \frac{4}3 g \sqrt{ \frac 1{2k} + \left( \frac m n \right)^2} \norm*{O}_{\operatorname{op}}^2.
\end{align*}
where we used \cref{lem:difftwist} and $\norm{\rho_\beta^{1/2} O(t)}_F \leq \norm{O(t)}_{\operatorname{op}} = \norm{O}_{\operatorname{op}}$.
If $O$ is a Pauli operator, then~$\norm{O}_{\operatorname{op}} = 1$, hence provided that $g^2 \ll k$, $gm \ll n$ we obtain
\begin{align}\label{eq:qtildeapp}
  \tilde q_{l_0}(g) \approx e^{-ig} \bra{\TFD} O_R(t) e^{igV} O_L^T(-t) \ket{\TFD}.
\end{align}
This establishes \cref{eq:qtildemaintext} in the main text.

\section{Proof of Eqs.~(\ref{eq:fid_main_text}) and (\ref{eq:fid_main_text2})}\label{app:Fid_Form}
In this section we prove our fidelity formulas for the state transfer protocol.
Formally, the state transfer protocol amounts to the channel
\begin{align*}
  \Psi_\text{in} \mapsto
  \Psi_\text{out}
= \tr_{A_\text{in}LB_R} \left[ U_R e^{igV} U^T_L \FF_{A_\text{in}A_L} \overline U_L ( \Psi_\text{in} \ot \proj{\TFD}_{LR} ) U^T_L \FF_{A_\text{in}A_L} \overline U_L e^{-igV} U_R^\dagger \right].
\end{align*}
Here, the time evolution is given by~$U = e^{-iH_R t} = \left(e^{-i H_L t}\right)^T$, and we initially place the $m$-qubit input state into an auxiliary Hilbert space~$A_\text{in}$, which is then inserted into the left message subsystem~$A_L$ at time $-t$ by using the swap operator $\FF_{A_\text{in}A_L}$ (compare (\cref{fig:Wormhole_Circuit}, left).

In view of \cref{eq:TS_Explicit}, we expect the Pauli operator $Y_{A_R}:=Y^{\ot m}$ to be a good decoding of the message.
The following lemma bounds the entanglement fidelity of the corresponding channel.


\begin{lem}\label{lem:fid exact}
The entanglement fidelity of the channel $\mathcal C(\Psi_\text{in}) = Y_{A_R} \Psi_\text{out} Y_{A_R}$ is given by
\begin{align*}
  F = \norm*{\EE_{P_A} (-1)^{\abs{P_A}} P_A(t) \left(\rho_\beta^{1/2} P_A(t) \right)^{(g)_X^B}}_F,
\end{align*}
where the average is over random Pauli operators~$P_A$ on $A$, we denote $P_A(t) = U^\dagger P_A U$, and the notation $O^{(g)_X^B}$ is defined in \cref{eq:X B twist}.
\end{lem}
\begin{proof}
By definition of the entanglement fidelity,
\begin{align*}
  F^2 = \braket{\phi^+_{A_RE} | Y_{A_R} \tr_{A_\text{in}LB_R} \left[ U_R e^{igV} U^T_L \FF_{A_\text{in}A_L} \overline U_L ( \phi^+_{A_\text{in}E} \ot \proj{\TFD}_{LR} ) U^T_L \FF_{A_\text{in}A_L} \overline U_L e^{-igV} U_R^\dagger \right] Y_{A_R} | \phi^+_{A_RE}}.
\end{align*}
Rewriting swap operator as $\mathbb F_{A_\text{in}A_L} = 2^{-m} \sum_{\vv} P_{A_\text{in}}^\vv P_{A_L}^\vv$, we get
\begin{align*}
  F^2
&= \sum_{\vv} \braum_{A_RE} Y_{A_R} \tr_{A_\text{in}LB_R} \Bigl[ U_R e^{igV} U^T_L \FF_{A_\text{in}A_L} \overline U_L ( \phi^+_{A_\text{in}E} \ot \proj{\TFD}_{LR}) \\
& \qquad\qquad\qquad U^T_L \FF_{A_\text{in}A_L} \overline U_L e^{-igV} U_R^\dagger \Bigr] Y_{A_R} \ketum_{A_RE} \\
&= 4^{-m} \sum_{\vv,\vw} \braum_{A_RE} Y_{A_R} \tr_{A_\text{in}LB_R} \Bigl[ U_R e^{igV} U^T_L P_{A_\text{in}}^\vv P_{A_L}^\vv \overline U_L ( \phi^+_{A_\text{in}E} \ot \proj{\TFD}_{LR} ) \\
& \qquad\qquad\qquad U^T_L P_{A_\text{in}}^\vv P_{A_L}^\vv \overline U_L e^{-igV} U_R^\dagger \Bigr] Y_{A_R} \ketum_{A_RE} \\
&= 4^{-m} \sum_{\vv,\vw} \braum_{A_RE} P_{A_R}^\vv Y_{A_R} \tr_{A_\text{in}LB_R} \Bigl[ U_R e^{igV} U^T_L P_{A_L}^\vv \overline U_L ( \phi^+_{A_\text{in}E} \ot \proj{\TFD}_{LR} ) \\
& \qquad\qquad\qquad U^T_L P_{A_L}^\vw \overline U_L e^{-igV} U_R^\dagger \Bigr] Y_{A_R} P_{A_R}^\vw \ketum_{A_RE} \\
&= 16^{-m} \sum_{\vv,\vw} \Tr\left[ P_{A_R}^\vv Y_{A_R} U_R e^{igV} U^T_L P_{A_L}^\vv \overline U_L \proj{\TFD}_{LR} U^T_L P_{A_L}^\vw \overline U_L e^{-igV} U_R^\dagger Y_{A_R} P_{A_R}^\vw\right]
\end{align*}
We continue using the definition of $\ket{\TFD} := 2^{n/2}(\rho_\beta^{1/2})_R\ketum_{LR}$, then the transpose trick, and finally \cref{eq:exact action}, which show that
\begin{align*}
&\quad P_{A_R}^\vv Y_{A_R} U_R e^{igV} U^T_L P_{A_L}^\vv \overline U_L \ket{\TFD}_{LR} \\
&= 2^{n/2} P_{A_R}^\vv Y_{A_R} U_R e^{igV} \left(\rho_\beta^{1/2}\right)_R U^T_L P_{A_L}^\vv \overline U_L \ketum_{LR} \\
&= 2^{n/2} \left( P_{A_R}^\vv Y_{A_R} U_R \left(\rho_\beta^{1/2} U_R^\dagger (P_{A_R}^\vv)^T U_R \right)^{(g)_X^B} \right) \ketum_{LR}.
\end{align*}
Thus we obtain
\begin{align*}
  F^2
&= 16^{-m} \sum_{\vv,\vw} \Tr\Biggl[ \left( P_{A_R}^\vv Y_{A_R} U_R \left(\rho_\beta^{1/2} U_R^\dagger (P_{A_R}^\vv)^T U_R \right)^{(g)_X^B} \right) \\
&\qquad\qquad\qquad\qquad \left( P_{A_R}^\vw Y_{A_R} U_R \left(\rho_\beta^{1/2} U_R^\dagger (P_{A_R}^\vw)^T U_R \right)^{(g)_X^B} \right)^\dagger \Biggr] \\
&= 16^{-m} \sum_{\vv,\vw} \Tr\Biggl[ \left( U_R^\dagger Y_{A_R} P_{A_R}^\vv Y_{A_R} U_R \left(\rho_\beta^{1/2} U_R^\dagger (P_{A_R}^\vv)^T U_R \right)^{(g)_X^B} \right) \\
&\qquad\qquad\qquad\qquad \left( U_R^\dagger Y_{A_R} P_{A_R}^\vw Y_{A_R} U_R \left(\rho_\beta^{1/2} U_R^\dagger (P_{A_R}^\vw)^T U_R \right)^{(g)_X^B} \right)^\dagger \Biggr] \\
&= 16^{-m} \Tr\Biggl[ \left( \sum_{\vv} (-1)^{\abs\vv} U_R^\dagger P_{A_R}^\vv U_R \left(\rho_\beta^{1/2} U_R^\dagger P_{A_R}^\vv U_R \right)^{(g)_X^B} \right) \\
&\qquad\qquad\qquad\qquad \left( \sum_{\vw} (-1)^{\abs\vw} U_R^\dagger P_{A_R}^\vw U_R \left(\rho_\beta^{1/2} U_R^\dagger P_{A_R}^\vw U_R \right)^{(g)_X^B} \right)^\dagger \Biggr] \\
&= \Tr\Biggl[ \left( \EE_{P_A} (-1)^{\abs{P_A}} P_A(t) \left(\rho_\beta^{1/2} P_A(t) \right)^{(g)_X^B} \right) \left( \EE_{P'_A} (-1)^{\abs{P'_A}} P'_A(t) \left(\rho_\beta^{1/2} P'_A(t) \right)^{(g)_X^B} \right)^\dagger \Biggr] \\
&= \norm*{\EE_{P_A} (-1)^{\abs{P_A}} P_A(t) \left(\rho_\beta^{1/2} P_A(t) \right)^{(g)_X^B}}_F^2.
\end{align*}
after inserting $Y_{A_R}^2 = I_{A_R}$, $U_R U_R^\dagger = I$, and using \cref{eq:Pauli_Transpose,eq:Ytranspose_Pauli}.
\end{proof}

We now use \cref{lem:fid exact} to derive bounds on the entanglement fidelity.
By the Cauchy-Schwarz inequality, using that $\norm{\rho_\beta^{1/2}}_F^2 = \norm{\rho_\beta}_{\operatorname{tr}} = 1$, we obtain the following lower bound:
\begin{align}\label{eq:Fid_Bounds_1}
  \abs*{\EE_{P_A} (-1)^{\abs{P_A}} \Tr\left[\rho_\beta^{1/2} P_A(t) \left(\rho_\beta^{1/2} P_A(t) \right)^{(g)_X^B}\right]}
\leq 
F.
\end{align}
Next, we give two useful upper bounds.
The first one follows by a simple triangle inequality:
\begin{align*}
  F
&= \norm*{\EE_{P_A} (-1)^{\abs{P_A}} Q}_F \\
&\leq \norm*{\EE_{P_A} (-1)^{\abs{P_A}} (Q - \tr[\rho_\beta^{1/2} Q] \rho_\beta^{1/2})}_F
+ \norm*{\EE_{P_A} (-1)^{\abs{P_A}} \tr[\rho_\beta^{1/2} Q] \rho_\beta^{1/2}}_F
\end{align*}
where we denote $Q = Q_{P_A} = P_A(t) \left(\rho_\beta^{1/2} P_A(t) \right)^{(g)_X^B}$.
Using $\norm Q_F=1$, the first term can be bounded as
\begin{align*}
  \norm*{\EE_{P_A} (-1)^{\abs{P_A}} (Q - \tr[\rho_\beta^{1/2} Q] \rho_\beta^{1/2})}_F
&\leq \EE_{P_A} \norm*{Q - \tr[\rho_\beta^{1/2} Q] \rho_\beta^{1/2}}_F
= \EE_{P_A} \sqrt{ 1 - \abs{\tr[\rho_\beta^{1/2} Q]}^2 },
\end{align*}
while the second term equals
\begin{align*}
  \norm*{\EE_{P_A} (-1)^{\abs{P_A}} \overline{\tr[\rho_\beta^{1/2} Q]} \rho_\beta^{1/2}}_F
= \abs*{\EE_{P_A} (-1)^{\abs{P_A}} \tr[\rho_\beta^{1/2} Q]} \, \norm*{\rho_\beta^{1/2}}_F
= \abs*{\EE_{P_A} (-1)^{\abs{P_A}} \tr[\rho_\beta^{1/2} Q]}.
\end{align*}
Accordingly, we obtain the upper bound
\begin{equation}\label{eq:good_upper_bound}
\begin{aligned}
F &\leq \abs*{\EE_{P_A} (-1)^{\abs{P_A}} \Tr\left[\rho_\beta^{1/2} P_A(t) \left(\rho_\beta^{1/2} P_A(t) \right)^{(g)_X^B}\right]} \\
&\quad + \EE_{P_A} \sqrt{1 - \abs*{\Tr\left[\rho_\beta^{1/2} P_A(t) \left(\rho_\beta^{1/2} P_A(t) \right)^{(g)_X^B}\right]}^2},
\end{aligned}
\end{equation}
which complements \cref{eq:Fid_Bounds_1}.
A second, simple upper bound on the entanglement fidelity can be obtained as follows:
\begin{align*}
  F^2
&= \EE_{P_A,P'_A} (-1)^{\abs{P_A}} (-1)^{\abs{P'_A}} \Tr[Q_{P_A} Q_{P'_A}^\dagger] \\
&
\leq \left( 1 - \frac2{4^m} \right)
+ \frac{2}{4^m} \abs*{\EE_{P_A} (-1)^{\abs{P_A}} \Tr[Q_{P_A} Q_{I_A}^\dagger]} \\ &=
1 - \frac2{4^m} \left(1 - \abs*{\EE_{P_A} (-1)^{\abs{P_A}} \Tr[Q_{P_A} Q_{I_A}^\dagger]} \right)
\end{align*}
where the first inequality follows by upper bounding all terms by 1 except for those where~$P_A$ or~$P'_A$ is the identity operator.
Thus we obtain the following upper bound, which is most useful for~$m=1$:
\begin{align}\label{eq:ex_upper}
  F \leq
1 - \frac1{4^m} \left(1 - \abs*{\EE_{P_A} (-1)^{\abs{P_A}} \Tr[Q_{P_A} Q_{I_A}^\dagger]} \right)
\end{align}

We finally evaluate these bounds in terms of the Fourier transform of the winding size distribution.
As in the main text, we assume that the winding size distribution of a thermal Pauli operator~$\rho_\beta^{1/2} P_A(t)$ only depends on the initial size~$\abs{P_A} = l$;
accordingly, we denote the Fourier transform by~$\tilde q_l(g)$.
Assuming the thermal operators satisfy the hypotheses of \cref{lem:difftwist} and $g^2 \ll k$, $gm \ll n$, we have
\begin{align*}
  \tilde q_l(g)
\approx e^{-ig} \Tr\left[ \rho_\beta^{1/2} P_A(t) \left( \rho_\beta^{1/2} P_A(t) \right)^{(g)_X^B} \right].
\end{align*}
(see \cref{eq:fourier defined as trace,eq:two point,eq:qtildeapp}).
Thus, the quantity $F_q$ defined in \cref{eq:def F_q main} in the main text can be computed as
\begin{align*}
  F_q
= \abs*{\EE_{P_A} (-1)^{\abs{P_A}} \tilde q_{\abs{P_A}}(g)}
  \approx \abs*{\EE_{P_A} (-1)^{\abs{P_A}} \Tr\left[\rho_\beta^{1/2} P_A(t) \left(\rho_\beta^{1/2} P_A(t) \right)^{(g)_X^B}\right]}.
\end{align*}
Thus the lower bound in \cref{eq:Fid_Bounds_1} and the upper bound in \cref{eq:good_upper_bound} become
\begin{align*}
  F_q \lessapprox F \lessapprox F_q + \sum_{l=0}^m (N_l/4^m) \sqrt{1 - \abs{\tilde q_l(g)}^2}.
\end{align*}
This proves \cref{eq:fid_main_text}.
At last, we evaluate \cref{eq:ex_upper} under the assumption that the width of the size distribution of the thermal state is small and of order $\sqrt n$.
(This is a common feature, with the important exception of completely nonlocal random Hamiltonian evoluations (GUE/GOE).)
In this case, $Q_I = (\rho_\beta^{1/2})^{(g)}$ is close to $\rho_\beta^{1/2}$, up to a global phase, hence $\tr[Q_P Q_I^\dagger] \propto \tilde{q}_{\abs P}(g)$, again up to a global phase, and we obtain the following bound:
\begin{align*}
  F
\lessapprox
1 - \frac1{4^m} \left( 1 - \abs*{\EE_{P_A} (-1)^{\abs{P_A}} \tilde{q}_{\abs{P_A}}(g)} \right)
= 1 - \frac1{4^m} \left( 1 - F_q \right).
\end{align*}
This establishes \cref{eq:fid_main_text2}.

\end{widetext}
\end{document}